\documentclass[11pt]{article}

\usepackage{ams}
\usepackage{amssymb}
\usepackage{graphicx}
\usepackage{caption}
\usepackage{pgf}
\usepackage{tikz}
\usepackage{amsmath}
\usepackage{amsthm}
\usepackage{amsfonts}

\usetikzlibrary{arrows,automata}

\usepackage{latexsym}
\usepackage{amsfonts}
\usepackage[english]{babel}
\usepackage{xy}
\usepackage{fancyhdr}
\usepackage[a4paper,top=3cm,bottom=3.5cm,left=2.3cm,right=2.3cm,%
bindingoffset=5mm]{geometry}

\usepackage{setspace}

\newtheorem{lemma}{Lemma}

\newtheorem{theorem}{Theorem}
\newtheorem{proposition}{Proposition}

\newcommand\independent{\protect\mathpalette{\protect\independenT}{\perp}}
\def\independenT#1#2{\mathrel{\rlap{$#1#2$}\mkern2mu{#1#2}}}

\newcommand{\Hv}{\mbox{\boldmath $ H $}}

\newcommand{\pigrecov}{\mbox{\boldmath $ \pi $} }

\allowdisplaybreaks

\begin{document}

$ \, $

\bigskip

$ \, $

\centerline{\Large \noindent {\bf Estimation of distributional effects  of treatment and control}}

\centerline{\Large \noindent {\bf under selection on observables:}}

\centerline{\Large \noindent {\bf consistency, weak convergence, and applications}}

\bigskip

\centerline{Pier Luigi Conti\footnote{{Pier
Luigi Conti. Dipartimento di Scienze
Statistiche; Sapienza Universit\`{a} di Roma; P.le A.
Moro, 5; 00185 Roma; Italy. E-mail
pierluigi.conti@uniroma1.it}}}

\centerline{Livia De Giovanni\footnote{{Livia De Giovanni.
Dipartimento di Scienze
Polituiche; LUISS Guido Carli; Viale Romania, 32; 00197 Roma; Italy. E-mail
ldegiovanni@luiss.it}}}




\bigskip

\centerline{\noindent {\bf Abstract}}

In this paper the estimation of the distribution function for potential outcomes to receiving or not receiving a treatment is studied. The approach is based on weighting observed data on the basis on the estimated propensity score. A weighted version of the empirical process is constructed and its weak convergence to bivariate Gaussian process is established. Results for the estimation of the Average Treatment Effect (ATE) and Quantile Treatment Effect (QTE) are obtained as by-products. Applications to the construction of nonparametric tests for the treatment effect and for the stochastic dominance of the treatment over control are considered, and their finite sample properties and merits are studied via simulation.

\bigskip

\noindent
{\bf Keywords}. Potential outcomes, Propensity score, causality, Empirical processes, weak convergence, nonparametric tests, stochastic dominance.

\bigskip

\newpage

\section{Introduction}
\label{sec:introduction}

The evaluation of the possible effects of a treatment on an outcome plays a central role in theoretical as well as applied statistical and econometrical literature; cfr. the excellent review papers by  \cite{AtheyImbens2017} and \cite{imbensjel}. The main quantity of interest, traditionally, is the average effect of the treatment on outcome, or better the difference between the expected valued of outcomes for treated and control (untreated) subjects, {\em i.e.} $ATE$ (Average Treatment Effect). Another quantity of interest is
the effects of treatment on outcome quantiles, which is summarized by $QTE$ (Quantile Treatment Effect).
The main source of difficulty is that data are usually observational, so that the estimation of the treatment effect by simply comparing outcomes for
treated vs. control subjects is prone to a relevant source of bias: receiving a treatment is not a ``purely random'' event, and there could be relevant
differences between treated and control subjects. This motivates the need to account for confounding covariates.

In the literature, several different techniques have been proposed to estimate $ATE$, under various assumptions (see \cite{AtheyImbens2017}, \cite{imbensjel} and references therein).
As far as $QTE$ is concerned, cfr. the paper by \cite{firpo07}. The problem of evaluating possible differences in the distribution function of potential outcomes with binary instrumental variables is studied in \cite{Abadie2002} via a Kolmogorv-Smirnov type test.

In the present paper we essentially focus on evaluating the possible effects of the treatment on the whole outcome probability distribution.
The starting point is to use outcome weighting similar to those introduced in \cite{HIR03} and \cite{firpo07}. Using this approach, estimates of the distribution function (d.f.)
for treated and control subjects will be obtained. Such estimators essentially play a role similar to the empirical d.f. in nonparametric statistics. It will be shown that the resulting  ``empirical processes'' weakly converge to an appropriate Gaussian process. Although it is non a Brownian bridge, it possesses
several properties  similar to the Brownian bridge (continuity of trajectories, etc.). These theoretical results are applied to the construction of confidence bands for the outcome distribution under treatment and under control, as well as to construct a new statistical test to compare treated and untreated subjects. In a sense, such a test is
a version of the classical Wilcoxon-Mann-Whitney test for two groups comparison. Its main merit is to capture the possible difference between treated and untreated subjects even when $ATE$ is equal to zero. Another application of interest will be the construction of a test
for stochastic dominance of treatment w.r.t. control, which is of interest, for instance, in programme evaluation exercises (\cite{linton05}), welfare outcome, etc..

The paper is organized as follows. In Section $\ref{sec:problem}$ the problem is described. In Section $\ref{sec:F1F0asymp}$  the main asymptotic large sample results are provided, and in Section $\ref{subsampling_general}$ approximations based on subsampling are considered. Particularizations to $ATE$ and
$QTE$ are given in Section $\ref{ATEandQTE}$. Section $\ref{sec:CB}$ is devoted to the construction of confidence bands for the d.f. of outcomes, for both
treated and untreated subjects. In Section $\ref{sec:Wilcoxon}$ a Wilcoxon-type statistic to test for treatment effect of the d.f of outcomes in introduced, and
in Section $\ref{sec:dominance}$ an elementary test for first-order stochastic dominance of treated {\em vs.} untreated is studied. The finite sample
performance of the proposed methodologies is studied {\em via} Monte Carlo simulation in Section $\ref{sec:simulation}$.

\section{The problem}
\label{sec:problem}

Let $Y$ be an outcome of interest, observed on a sample of subjects. Some of the sample units are treated with an appropriate treatment (treated group);
the other sample units are untreated (control group).
If $T$ denotes the treatment indicator variable, then whenever $T= 1$,  $Y_1$ is observed; otherwise, if $T= 0$, $Y_0$ is observed. Here $Y_1$ and $Y_0$ are the {\em potential outcomes} due to receiving and not receiving the treatment, respectively.
The observed outcome is then equal to $Y= TY_{(1)} +  (1 - T) Y_{(0)}$.
In the sequel,  $F_1(y)=P(Y_{(1)} \le y)$ will denote the distribution function (d.f.) of $Y_{(1)}$, and $F_0(y)=P(Y_{(0)} \le y)$ the d.f. of $Y_{(0)}$.

As already said in the introduction, receiving a treatment is not a ``purely random'' event, as in experimental framework. On the contrary, there could be relevant differences between treated and untreated subjects, due to the presence of confounding covariates. In the sequel, we will denote by $X$ the (random)
vector of relevant covariates, that is assumed to be observed.

In order to get consistent estimates, identification restrictions are necessary. The relevant restriction
assumed in the sequel is selection of treatment is based on observable variables: given a set of observed covariates, assignment either to the treatment group or to the control group is random.  Formally speaking, let $p(x) = P(T = 1 \vert X = x)$ be
the conditional probability of receiving the treatment given covariates $X$; it is termed {\em propensity score}. The marginal probability of being treated, $P(T = 1)$, is equal to $E[p(X)]$.

In the sequel, our main assumption is that the {\em strong ignorability} conditions (cfr. \cite{rosrub83}) are fulfilled. In more detail, consider next the joint distribution of  ($Y_{(1)}, \, Y_{(0)}, \, T, \, X$), and denote by $\mathcal{X}$ the support of $X$.
The following assumptions are assumed to hold.
\begin{itemize}
\item[(i)] Unconfoundedness (cfr. \cite{rubin77}): given $X$, $(Y_{(1)}, \, Y_{(0)})$ are jointly independent of $T$: $( Y_{(1)}, \, Y_{(0)}) \independent T \vert X$.
\item[(ii)] The support of $X$, $\mathcal{X}$ is a compact subset of $\mathbb{R}^l$.
\item[(iii)] Common support: there exists $\delta>0$ for which $\delta \le p(x) \le 1-\delta \; \forall \, x \in \mathcal{X}$, so that $\underset{x}{\inf} \, p(x) \ge \delta$, $\underset{x}{\sup} \, p(x) \le 1-\delta$.
\end{itemize}
Assumption $(i)$ is also known as \textit{Conditional Independence Assumption} ($CIA$).

For the sake of simplicity, we will use in the sequel the notation
\begin{eqnarray}
p_{1} ( x ) =  p ( x ) , \; p_0 ( x ) =  1- p ( x ) . \label{eq:simp_01}
\end{eqnarray}
\noindent From the above assumptions, the basic relationships
\begin{eqnarray}
E \left [ \frac{1}{p_{j} (x)} I_{(T=j)}  I_{(Y \le y)} \right ] & = & E_x \left [ E \left [ \left . \frac{1}{p_{j} (x)} I_{(T=j)}
I_{(Y_{(j)} \le y)} \right \vert  x \right ]
\right ] \nonumber \\
\, & = & E_x \left [ \frac{1}{p_{j} (x)} E \left [ \left . I_{(T=j)} \right \vert x \right ]
E \left [ \left . I_{( Y_{(j)} \le y)} \right \vert x \right ] \right ] \nonumber \\
\, & = &  E_x \left [ F_j (y \vert x)) \right ] \nonumber \\
\, & = & F_j(y) , \;\; j=1, \, 0.
 \label{eq:e1}
\end{eqnarray}
\noindent are obtained.

The {\em Average Treatment Effect} ({\em ATE}) is defined as $\tau = E[ Y_{(1)} ] - E[ Y_{(0)} ]$. The estimation of {\em ATE} is a problem of primary importance in the literature,
and several different approaches have been proposed (\cite{AtheyImbens2017} and references therein). Another parameter of interest in the {\em Quantile Treatment
Effect} ({\em QTE}), which is the difference between quantiles of $F_1$ and $F_0$: $F_1^{-1} (p) - F_0^{-1} (p)$, with $0<p<1$; cfr.
 \cite{firpo07}. In particular, when $p=1/2$ it
reduces to the {\em Median Treatment Effect}.

As already said in the introductory section, in the present paper we concentrate on the estimation of the d.f.s  $F_1 (y)$, $F_0 (y)$
under treatment and control, respectively.
As special cases, the results in \cite{HIR03} and \cite{firpo07} will be obtained.

\section{Estimation of $F_1, F_0$}
\label{sec:F1F0est}

\subsection{Basics}
\label{sec:F1F0est_def}

The basic approach to the estimation of $F_1$, $F_0$ follows, in principle, the ideas developed in \cite{HIR03} to estimate {\em ATE}. First of all,
the propensity score $p(x)$ is estimated by a sieve estimator $\widehat{p}_n (x)$, say; cfr. \cite{HIR03}, \cite{firpo07}. Let $\Hv_K (x) =
\{ H_{k,j} (x) \}$, $j=1, \, \dots , \, K$ be a $K$-dimensional vector of polynomials in $x \in  \mathcal{X}$, such that
\begin{itemize}
\item[S1.] $\Hv_K ; \mathcal{X} \rightarrow \mathbb{R}^K$;
\item[S2.] $H_{k,1} (x) =1 $;
\item[S3.] $\Hv_K$ includes all polynomials up to order $n$ whenever $K > (n+1)^r$, with $K = K(n) \rightarrow \infty$ as $ n
\rightarrow \infty$.
\end{itemize}
The propensity score is approximated by a linear combination of $ H_{k,j} (x)$ on a logit scale, with coefficients estimated
by maximizing a pseudo-likelihood. More formally, if $L(z) =1 / ( 1+ e^{-x} )$, then
$ \widehat{p}_n (x) = L( \Hv_K (x)^{T} \widehat{\pigrecov}_K )$,
where the $K$-dimensional vector $\widehat{\pigrecov}_K$ is estimated
by maximum likelihood method:
\begin{eqnarray}
\widehat{\pigrecov}_K = {\mathrm{argmax}} \frac{1}{n} \sum_{i=1}^{n}
\left \{ T_i \log \left ( L(  \Hv_K (x)^{T} \pigrecov_K ) \right ) + ( 1- T_i )
\log \left ( L( 1- \Hv_K (x)^{T} \pigrecov_K ) \right )
\right \} . \nonumber
\end{eqnarray}
In the sequel, the following result will be widely used.

\begin{theorem}
\label{th_consist_sieve}
Assume that S1 - S3 are fulfilled, and that $p(x)$ is continuously differentiable of order $s \geq 7 l$, with $l= {\mathrm{dim}} ( \mathcal{X} )$. If
$K= n^{\nu}$, with $ 1/ ( 4 ( s/l -1 )) < \nu < 1/9$, then
\begin{eqnarray}
\sup_x \left \vert \widehat{p}_n (x) - p(x) \right \vert \stackrel{p}{\rightarrow} 0 \;\; {\mathrm{as}} \; n \rightarrow \infty .
\label{eq:consist_sieve}
\end{eqnarray}
\end{theorem}
\noindent {\bf Proof.} See \cite{HIR03}.

Again, for notational simplicity, and similarly to $( \ref{eq:simp_01} )$, define:
\begin{eqnarray}
\widehat{p}_{1,n} ( x ) =  \widehat{p}_{n} ( x ) , \; \widehat{p}_{0,n} ( x ) =  1- \widehat{p}_{n} ( x ) .
\label{eq:simplif_not}
\end{eqnarray}

In order to estimate $F_1$ and $F_0$, the following ``H\'{a}jek - type'' estimators are considered:
\begin{eqnarray}
\widehat{F}_{1,n}(y) = \sum_{i=1}^{n} w^{(1)}_{i,n} I_{(Y_i \le y)} , \;\;\;
\widehat{F}_{0,n}(y) = \sum_{i=1}^{n} w^{(0)}_{i,n} I_{(Y_i \le y)}
\label{eq:e5}
\end{eqnarray}
\noindent where
\begin{eqnarray}
w^{(j)}_{i,n}=\frac{I_{(T_i=j)}/\widehat{p}_{j,n} (x_i) }{\sum_{k=1}^{n} I_{(T_k=1)}/\widehat{p}_{j,n} (x_k) } , \; j=1, \, 0;
 \;\; i=1, \, \dots , \, n .
\label{eq:e6}
\end{eqnarray}
It is immediate to see that $( \ref{eq:e5} )$ are proper d.f.s, {\em i.e.} they are {\em bona fide} estimators.

As alternative estimators of $F_1$, $F_0$, the following ``Horvitz-Thompson - type'' estimators could be considered:
\begin{eqnarray}
\widehat{F}^{HT}_{1,n}(y) =  \frac{1}{n} \sum_{i=1}^{n} \frac{I_{(T_i=1)}}{\widehat{p}_{1,n}(x_i)} I_{(Y_i \le y)} , \;\;\;
\widehat{F}^{HT}_{0,n}(y) = \frac{1}{n} \sum_{i=1}^{n} \frac{I_{(T_i=0)}}{\widehat{p}_{0,n} (x_i)}   I_{(Y_i \le y)} .
\label{eq:e5b}
\end{eqnarray}
We will mainly concentrate on $( \ref{eq:e5} )$ for two reasons. First of all, $( \ref{eq:e5b} )$ are not proper d.f.s, because
$\widehat{F}^{HT}_{1,n}( + \infty ) \neq 1$, $\widehat{F}^{HT}_{0,n}( + \infty ) \neq 1$ with positive probability. In the second place, as it will
be seen in the sequel, $( \ref{eq:e5b} )$ are asymptotically equivalent to $( \ref{eq:e5} )$.

\subsection{Basic asymptotic results}
\label{sec:F1F0asymp}

The goal of the present section is to study the asymptotic, large sample, properties of estimators $( \ref{eq:e5} )$. Our first result
is a Glivenko - Cantelli type result, showing the uniform consistency (in probability) of $F_{1,n}(y)$, $F_{0,n}(y)$.

\begin{proposition}
\label{gliv-cant}
Assume that the conditions of Th. $\ref{th_consist_sieve}$ are fulfilled.
Then:
\begin{eqnarray}
\sup_{x} \left \vert \widehat{F}_{1,n}(y) - F_1 (y) \right \vert \overset{p}{\rightarrow} 0 , \;\;
\sup_{x} \left \vert \widehat{F}_{0,n}(y) - F_0 (y) \right \vert \overset{p}{\rightarrow} 0\;\;\; {\mathrm{as}} \; n \rightarrow \infty .
\label{eq:glivenko_cantelli}
\end{eqnarray}
\end{proposition}
\noindent {\bf Proof.} See Appendix.

Next step consists in studying the limit, large sample distribution of the above estimators. Define first the stochastic process
\begin{eqnarray}
W_n (y)  =
\left[ \begin{array}{cc}
W_{1,n} (y) \\
W_{0,n} (y)
\end{array} \right]
= \left[ \begin{array}{cc}
\sqrt{n} (\widehat{F}_{1,n}(y)-F_1(y)) \\
\sqrt{n} (\widehat{F}_{0,n}(y)-F_0(y))
\end{array} \right] , \;\; y \in \mathbb{R}
\label{eq:emp_proc}
\end{eqnarray}
The bivariate stochastic process $W_n ( \cdot )$  $( \ref{eq:emp_proc} )$ essentially plays the same role as the empirical process in
classical non-parametric statistics, with a complication due to the presence of $\widehat{F}_{1,n}(y)$, $\widehat{F}_{0,n}(y)$ instead of
the usual empirical distribution function.

The weak convergence of $W_n ( \cdot )$ can be proved similarly to the classical empirical process, with modifications.
In the first place, from
\begin{eqnarray}
\sqrt{n} (\widehat{F}_{j,n}(y)-F_j(y)) = \left ( \frac{1}{n}\sum_{i=1}^{n} \frac{I_{(T_i=j)}}{\widehat{p}_{j,n} (x_i)} \right )^{-1}
\frac{1}{\sqrt{n}} \sum_{i=1}^{n} \frac{I_{(T_i=j)}}{\widehat{p}_{j,n} (x_i)}(I_{(Y_i \leq y)} - F_j (y)) , \;\; j=1, \, 0 \nonumber
\end{eqnarray}
\noindent and from Lemma \ref{lemma2} , it is seen that
the limiting distribution of $W_{n}(y)$, if it exists, coincides with the limiting distribution of
\begin{eqnarray}
\left[ \begin{array}{cc}
\frac{1}{\sqrt{n}} \sum_{i=1}^{n} \frac{I_{(T_i=1)}}{\widehat{p}_{1,n}(x_i)}(I_{(Y_i \le y)}-F_1(y)) \\
\frac{1}{\sqrt{n}} \sum_{i=1}^{n} \frac{I_{(T_i=0)}}{\widehat{p}_{0,n}(x_i)}(I_{(Y_i \le y)}-F_0(y))
\end{array} \right] , \;\; y \in \mathbb{R} .
\label{eq:emp_proc_interm}
\end{eqnarray}
In the second place, by repeating {\em verbatim} the arguments in Th. 1 in \cite{HIR03}, and \cite{HIRaddendum}, with $I_{(Y_{i} \leq y)}$ instead of $Y_i$ and
$F_j (y \vert x) = P( Y_{(j)} \leq y \vert x) $ instead of $E [ Y_{(j)} \leq y \vert x ]$, it is seen that, if
$K= n^{\nu}$, with $ 1/ ( 4 ( s/l -1 )) < \nu < 1/9$, then the relationship
\begin{eqnarray}
\left[ \begin{array}{cc}
\frac{1}{\sqrt{n}} \sum_{i=1}^{n} \frac{I_{(T_i=1)}}{\widehat{p}_{1,n}(x_i)}(I_{(Y_i \le y)}-F_1(y)) \\
\frac{1}{\sqrt{n}} \sum_{i=1}^{n} \frac{I_{(T_i=0)}}{\widehat{p}_{0,n}(x_i)}(I_{(Y_i \le y)}-F_0(y))
\end{array} \right]  =
\left[ \begin{array}{cc}
\frac{1}{\sqrt{n}} \sum_{i=1}^{n} Z_{1,i} (y) \\
\frac{1}{\sqrt{n}} \sum_{i=1}^{n} Z_{0,i} (y)
\end{array} \right]  + o_p (1)
, \;\; y \in \mathbb{R} .
\label{eq:emp_proc_interm2}
\end{eqnarray}
\noindent holds, where
\begin{eqnarray}
Z_{j,i} (y)= \left ( \frac{I_{(T_i=j)}}{p_{j} (x_i)}I_{(Y_i \le y)}-F_j(y) \right )-
\frac{F_{j} (y \vert x_i)}{p_{j} (x_i)} \left ( I_{(T_i=j)}-p_{j} (x_i) \right ) , \;\; j=1, \, 0; \; i=1, \, \dots , \, n .
\label{eq:def-z}
\end{eqnarray}
The term $o_p (1)$ appearing in $( \ref{eq:emp_proc_interm2} )$ depends on $y$, and, as it appears by using the bounds in  \cite{HIRaddendum},
convergence in probability to zero (or better, to the vector $[0, \, 0]^{T}$) holds uniformly over compact sets of $y$s. Hence, in order to prove that
the sequence of stochastic processes $( \ref{eq:emp_proc} )$ converges weakly to a limit process, it is enough to prove that $( \ref{eq:emp_proc_interm2} )$
converges weakly to a limiting process.

\begin{proposition}
\label{weak-conv}
Assume that the conditions of Th. $\ref{th_consist_sieve}$ are fulfilled, and that $F_1 (y)$, $F_1 (y \vert x ) $, $F_0 (y)$,
$F_0 (y \vert x )$ are continuous.
Then, the sequence of stochastic processes $( \ref{eq:emp_proc} )$ converges weakly, as $n$ goes to infinity, to a
Gaussian process $W(y) = [ W_1 (y) , \, W_0 (y) ]^{T}$ with null mean function ($E[ W_j (y) ] = 0$, $j=1, \, 0$) and covariance kernel:
\begin{eqnarray}
C(y,t) = E \left [ W(y) \otimes W(t) \right ] = \begin{bmatrix}
C_{11}(y,t) & C_{10}(y,t)\\
C_{01}(y,t) & C_{00}(y,t)\\
\end{bmatrix}
\end{eqnarray}
where:
\begin{eqnarray}
C_{jj} (y, \, t) & = & E \left [  \frac{1}{p_{j} (x)}(F_j {(y \land t \vert x )} - F_j (y \vert x) F_j (t \vert x)) \right ] \nonumber \\
& + &
E_x \left [ (F_j(y \vert x) - F_j(y)) (F_j ( t \vert x) - F_j(t)) \right  ] , \;\; j=1, \, 0 ;
\label{eq:e32} \\
C_{10} (y, \, t) & = & E \left [(F_1(y \vert x)-F_1(y))(F_0(t \vert x)-F_0(t)) \right ] \nonumber \\
\, & = &
E \left  [ F_1(y\vert x) F_0(t \vert x) \right ] - F_1(y) F_0(t) ;
\label{eq:e34} \\
C_{01}(y, \, t) & = & C_{10} (t, \, y) = E \left [ (F_1(t \vert x)-F_1(t))(F_0(y \vert x)- F_0(y)) \right ] .
\label{eq:e3b4}
\end{eqnarray}
Weak convergence takes place in the set $l_2^{\infty} ( \mathbb{R} )$ of bounded functions $\mathbb{R} \mapsto \mathbb{R}^2$ equipped
with the sup-norm (if $f = (f_1 , \, f_0 )$) $\| f \| = \sup_y \vert f_1 (y ) \vert + \sup_y \vert f_0 (y ) \vert $.
\end{proposition}
\noindent {\bf Proof.} See Appendix.

Due to the continuity of $F_1$, $F_0$, the weak convergence of Proposition $\ref{weak-conv}$ also holds in the space $D[-\infty, +\infty]^2$
of $\mathbb{R}^2$-valued  c\`{a}dl\`{a}g functions equipped with the Skorokhod topology.

Consider now the Horvitz-Thompson estimators $( \ref{eq:e5b} )$, and
define:
\begin{eqnarray}
W^{HT}_{jn} (y) = \sqrt{n} ( \widehat{F}^{HT}_{j,n}(y) - F_j (y) ), \;\; j=1, \, 0. \nonumber
\end{eqnarray}
From the proof of Proposition $\ref{weak-conv}$, it appears that the sequence of stochastic processes $W^{HT}_n ( \cdot )= [ W^{HT}_{1n} ( \cdot ) , \,
W^{HT}_{0n}  ( \cdot ) ]^{T}$ converges weakly to the same Gaussian limiting process  $W( \cdot ) = [ W_1 ( \cdot ) , \, W_0 ( \cdot ) ]^{T}$ that appears
in Proposition $\ref{weak-conv}$. Hence, the Horvitz-Thompson estimators $( \ref{eq:e5b} )$ are asymptotically equivalent to the
H\'{a}jek  estimators $( \ref{eq:e5} )$.

As well known, in classical nonparametric statistics the empirical process converges weakly to a Brownian bridge, on the scale of the population ditribution function. The limiting process $W( \cdot )$ in Proposition $\ref{weak-conv}$ is not a Browinian bridge, of course, although it is a Gaussian process. However, it shares with the Brownian bridge an important property: it possesses trajectories that are a.s. continuous.

\begin{proposition}
\label{continuous-trajectories}
If $F_0$ and $F_1$ are continuous, the limiting process $W ( \cdot ) = [ W_1 ( \cdot ) , \, W_0 ( \cdot ) ] $ possesses trajectories that are continuous with probability 1.
\end{proposition}
\noindent {\bf Proof.} See Appendix.

\subsection{Differentiable functionals}
\label{sec:hadamard}

The result of Proposition $\ref{weak-conv}$ can be immediately extended to general Hadamard differentiable functionals of $( F_1 , \, F_0)$, again assuming the continuity of $F_0$, $F_1$.
Consider a general functional:
\begin{eqnarray}
\theta = \theta(F_1, \, F_0): \; l^{\infty}(\mathbb{R})^2 \rightarrow \mathbb{E} \nonumber
\end{eqnarray}
\noindent where $l^{\infty}(\mathbb{R})^2$ is equipped with the $sup$-norm metric and $\mathbb{E}$ is a normed space
equipped with a norm $\| \cdot \|_\mathbb{E}$. As seen in Proposition $\ref{continuous-trajectories}$, the limiting process
$W ( \cdot ) = ( W_1 ( \cdot ), \, W_0 ( \cdot ) )$ concentrates on $C(\mathbb{\overline{R}})^2$, where
$C(\mathbb{\overline{R}} )$ is the set of continuous functions on the extended real line $\overline{\mathbb{R}}$. Note that functions in
$C(\mathbb{\overline{R}} )$ are bounded.

The functional  $\theta$ is Hadamard differantiable at $(F_1, \, F_0)$ tangentially to $C(\mathbb{\overline{R}})^2$ if there exists a linear application
\begin{eqnarray}
\theta^{\prime}_{(F_1, \, F_0)}: \; C(\mathbb{\overline{R}})\times C(\mathbb{\overline{R}}) \rightarrow \mathbb{E} \nonumber
\end{eqnarray}
\noindent  such that:
\begin{eqnarray}
\left \| \frac{\theta \left ( (F_1, \, F_0) + t h_{t} \right ) - \theta(F_1, \, F_0)}{t} - \theta^{\prime}_{(F_1,\, F_0)}(h)
\right \|_{\mathbb{E}} \rightarrow 0 \;\; {\mathrm{as}} \; t \downarrow 0, \; \forall \, h_t \rightarrow h . \nonumber
\end{eqnarray}
Using Theorem 20.8 in \cite{vandervaart98}, we then have:
\begin{eqnarray}
\sqrt{n}\left ( \theta(\widehat{F}_1, \, \widehat{F}_0)-\theta(F_1, \, F_0) \right ) \overset{d}{\rightarrow}
\theta^{\prime}_{(F_1, \, F_0 )}(W) .
\end{eqnarray}
In general, since  $\theta^{\prime}_{(F_1, \, F_0 )}(W)$ is a linear functional of a Gaussian process, it is a Gaussian process, as well. In particular,
if $\theta$ is a real-valued functional, then
$\theta^{\prime}_{(F_1, \, F_0 )}(W)$ has a Gaussian distribution with zero expectation and variance
\begin{eqnarray}
\sigma_{\theta}^2 = E \left [ \theta^{\prime}_{(F_1, \, F_0 )}(W)^2 \right ] .
\label{eq:asymp_var}
\end{eqnarray}
For the sake of simplicity, let $\widehat{\theta}_n$ be equal to $ \theta(\widehat{F}_1, \, \widehat{F}_0)$. The above result can be rewritten as
\begin{eqnarray}
\sqrt{n} \left ( \widehat{\theta}_n - \theta \right ) \stackrel{d}{\rightarrow} N(0,\sigma^2_{\theta}) \;\; {\mathrm{as}} \; n \rightarrow \infty \label{eq:asympt_t}
\end{eqnarray}
\noindent where the asymptotic variance $\sigma^2_{\theta}$ is given by $( \ref{eq:asymp_var} )$.

\section{Subsampling approximation}
\label{subsampling_general}

Consider a functional $\theta = \theta ( F_1 , \, F_0 )$. In order to construct a confidence interval on the basis of $( \ref{eq:asympt_t} )$, a consistent estimate of the asymptotic variance $\sigma_\theta^2$ $( \ref{eq:asymp_var} )$ is necessary. Unfortunately, apart a few cases, this is not simple, because
 $\sigma_\theta^2$ could
depend on $F_1$, $F_0$ in a complicate way, and a direct estimation could not be possible. This is the case, for instance, of quantiles, that will be
dealt with in next section. Here we briefly present a simple approach based on subsampling.

Define $A_i=(X_i, T_i, Y_i)$, $i=1, \, \dots , \, n$, and  consider all the ${n}\choose{m}$ subsamples of size $m$ of $(A_1, \, \dots , \, A_n)$.
Let further $\widehat{\theta}_{m,l}$ be the statistic $\widehat{\theta} (\cdot)$ computed for the $l$-th subsample of size $m$.
Next, consider then the empirical distribution function of the ${n}\choose{m}$ quantities $\sqrt{m} ( \widehat{\theta}_{m,l} -
\widehat{\theta}_n )$. In symbols:
\begin{eqnarray}
R_{n,m}(u) = {{n}\choose{m}}^{-1} \sum_{l=1}^{{n}\choose{m}} I_{\left ( \sqrt{m} ( \widehat{\theta}_{m,l} -
\widehat{\theta}_n) \le u \right )} .
\label{eq:e60}
\end{eqnarray}

If:
\begin{itemize}
\item[U1.] $\sqrt{n} (\widehat{\theta}_{n}-{\theta})\overset{d}{\rightarrow} N(0,\sigma^2_{\theta})$;
\item[U2.] $m$ depends on $n$ in such a way that $m \rightarrow \infty$, $\frac{m}{n} \rightarrow 0 \;\; {\mathrm{as}} \; n \rightarrow \infty$;
\end{itemize}
\noindent then, using Th. 2.1 in \cite{politrom94}, we have
\begin{eqnarray}
R_{n,m}(u)\overset{p}{\rightarrow} \Phi \left ( \frac{u}{\sigma_{\theta}} \right ) {\mathrm{as}} \; n, \, m \rightarrow \infty
\label{eq:e61}
\end{eqnarray}
\noindent where $\Phi$ is the distribution function of the Gaussian $N(0,1)$ distribution. The convergence in (\ref{eq:e61}) is uniform in $u$.

Relationship  $( \ref{eq:e61} )$ tells us that $Pr(\sqrt{n} (\widehat{\theta}_{n}-{\theta}) \le u)$
can be (uniformly) approximated by $R_{n,m}(u)$, as $n$ and $m$ get large.
From the continuity and strict monotonicity of $\Phi$, it follows that the empirical quantile $R^{-1}_{n,m}(p)= \inf \{ u: \; R_{n,m}(u) \geq 1 \}$
converges in probability to the quantile of order $p$ of the distribution $N(0,\sigma^2_{\theta})$ $\forall p \in (0,1)$.

The number of subsamples of size $m$, ${n}\choose{m}$ in $( \ref{eq:e60})$ can be very large, and then $R_{n,m}$ could be difficult  to be computed. In this case a ``stochastic'' version of $R_{n,m}$ can be considered according to the following steps.
\begin{itemize}
\item[1.] Select $M$ independent subsamples of size $m$ from $(A_1,\dots,A_n)$.
\item[2.] Compute the corresponding values $\widehat{\theta}_{m,1}, \, \dots, \,  \widehat{\theta}_{m,M}$ of the statistic $\widehat{\theta}$.
\item[3.] Compute of the corresponding empirical distribution function:
\begin{eqnarray}
\widehat{R}_{n,m}(u)=\frac{1}{M} \sum_{l=1}^{M} I_{\left ( \sqrt{m} ( \widehat{\theta}_{m,l} - \widehat{\theta}_n) \leq u \right )}  .
\label{eq:edf_subs}
\end{eqnarray}
\end{itemize}

It can be easily verified that if $M \rightarrow \infty$,  $n,m \rightarrow \infty$ and $\frac{m}{n} \rightarrow 0$, then $\widehat{R}_{n,m}(u)$ has the same limiting behaviour as $R_{n,m}(u)$.
These results can be used to obtain confidence intervals for $\theta$ and for testing statistical hypotheses {\em via} inversion of confidence intervals.
In more detail, let
\begin{eqnarray}
\widehat{R}_{n,m}^{-1} (u) = \inf \{ u: \; \widehat{R}_{n,m}(u) \geq p  \} \nonumber
\end{eqnarray}
\noindent be the $p$th quantile of $\widehat{R}_{n,m}$. It is easy to show that the interval:
\begin{equation}
\left [ \widehat{\theta}_{n} -\frac{1}{\sqrt{n}} \widehat{R}^{-1}_{n,m} \left (1- \frac{\alpha}{2} \right ), \;
\widehat{\theta}_{n} - \frac{1}{\sqrt{n}} \widehat{R}^{-1}_{n,m} \left ( \frac{\alpha}{2} \right ) \right ]
\label{eq:e62}
\end{equation}
is confidence interval for $\theta$ of asymptotic level $1-\alpha$.

The confidence interval $( \ref{eq:e62} )$ can be also used for testing the hypothesis:
\begin{eqnarray*}
\begin{cases}
H_0: \theta=\theta_0\\
H_1: \theta \ne \theta_0\\
\end{cases}
\end{eqnarray*}
If $\theta_0$ is in the confidence interval, then $H_0$ is accepted, otherwise it is rejected. Clearly, this is a test of asymptotic significance level $1-\alpha$.

\section{Average and Quantile Treatment Effect}
\label{ATEandQTE}

The results obtained so far allow one to re-obtain, as special cases, results previously obtained by \cite{HIR03} and \cite{firpo07}. They are presented below.

\subsection{Average Treatment Effect}

The Average Treatment Effect (ATE, for short) is defined as:
\begin{eqnarray}
\tau = E [ Y_{(1)} ] - E [ Y_{(0)} ] = \int_{\mathbb{R}}  y \, d [ F_1 (y) - F_0 (y) ] . \label{eq:tau}
\end{eqnarray}
In the sequel, we will assume that $E [ Y_{(1)}^2 ] $ and $ E [ Y_{(0)}^2 ]$ are both finite.
As an estimator of $\tau$, consider
\begin{eqnarray}
\widehat{\tau} & = & \int_{- \infty}^{+ \infty} y \, d [ \widehat{F}_{1,n}(y) -  \widehat{F}_{0,n}(y) ] \nonumber \\
\, & = & \sum_{i=1}^{n} y_i w^{(1)}_{i,n} - \sum_{i=1}^{n} y_i w^{(0)}_{i,n} . \label{eq:tauhat}
\end{eqnarray}
\noindent where the weights $w^{(j)}_{i,n}$, $j=1, \, 0$ are given by $( \ref{eq:e6})$.

As it appears from  $( \ref{eq:tau} )$, $\tau$ is a linear functional of $( F_1 , \, F_0 )$ and hence Hadamard differentiable. An integration by parts shows
that the asymptotic distribution of $\widehat{\tau}$ coincides with that
\begin{eqnarray}
- \int_{- \infty}^{+ \infty} \left ( W_1 (y) - W_0 (y)  \right ) \, dy
\nonumber
\end{eqnarray}
\noindent that turns out to normal with zero mean and variance
\begin{eqnarray}
\sigma^2_{\tau} = \int_{- \infty}^{+ \infty} \int_{- \infty}^{+ \infty} \left \{ C_{11} (y, \, t ) - C_{10} (y, \, t ) - C_{01} (y, \, t )
+ C_{00} (y, \, t ) \right \} \, dy \, dt . \nonumber
\end{eqnarray}
It is not difficult to see that the estimator $\widehat{\tau}$ $( \ref{eq:tauhat} )$ is asymptotically equivalent to that introduced in \cite{HIR03}.

\subsection{Quantiles and Quantile Treatment Effect}
\label{sec:quantiles}

Let $Q_j(p)=F_j^{-1}(p)= \inf \{ y: \; F_1(y) \geq p \}$, $0 < p < 1$ be the quantile of order $p$
of $F_j$, $j=1, \, 0$. In the sequel, we will assume that $Q_1(p)$, $Q_0(p)$ are in the common support of $F_1$, $F_0$. Furthermore, we will
denote by $supp (F_j )$ the support of $F_j$, $j=1, \, 0$.

Suppose that $F_1$, $F_0$ are continuous with positive density functions $f_1$, $f_0$, respectively:
\begin{eqnarray}
f_j(y) = \frac{dF_j(y)}{dy} >0 \;\; \forall \, y \in supp(F_j) , \;\; j=1, \, 0 .
\nonumber
\end{eqnarray}
As a consequence of the above assumption, $F_j$ is strictly monotonic (in its support).

Consider now $p_1, p_2$ ($0  <p_1 < p_2 <1$) such that $Q_1(p_1), Q_0(p_1), Q_1(p_2), Q_0(p_2)$ lie in the common support of $F_1$, $F_0$.
It is intuitive to estimate the quantile $Q_j (p)$ by its ``empirical counterpart''
\begin{eqnarray}
\widehat{Q}_{j,n}(p) = \widehat{F}_{j,n}^{-1}(p) = \inf \{y: \; \widehat{F}_{j,n}(y) \geq p \}, \;\; j=1, \, 0.
\label{eq:e36}
\end{eqnarray}

Let now $\mathbb{D}$ be the set of the restrictions of the distribution functions in $\mathbb{R}$ to $[a,b]$, and let $D[a,b]$ be
the set of c\`{a}dl\`{a}g  functions  in $[a,b]$. From \cite{vandervaart98}, it is seen  that the map $
G \longmapsto G^{-1} $  (from $\mathbb{D} \subseteq D[Q(p_1), Q(p_2)] $ onto $ l^\infty(0,1))$ is Hadamard differentiable at $( F_1, \, F_0 )$ tangentially to $C[a,b]$ with derivative:
\begin{eqnarray}
h \longmapsto -\biggl(\frac{h}{f} \biggr) \circ F^{-1} . \nonumber
\end{eqnarray}

Using then Th. 20.8 in \cite{vandervaart98}, (cfr. \cite{dossgill92} for an equivalent approach), the process
\begin{eqnarray}
\left[ \begin{array}{cc}
\sqrt(n)(\widehat{Q}_{1,n}(p)-Q_1(p) \\
\sqrt(n)(\widehat{Q}_{0,n}(p)-Q_0(p))
\end{array} \right] , \;\; p \in [p_1,p_2]
\end{eqnarray}
\noindent converges weakly as $n \rightarrow \infty$ (on $ l^\infty(p_1,p_2) $ equipped with the $sup$-norm)  to a Gaussian process $Z(p)=[Z_1(p),Z_2(p)]^{\prime}$ defined as:
\begin{eqnarray}
Z(p)  =
\left[ \begin{array}{cc}
-\frac{W_1(Q_1(p))}{f_1(Q_1(p))}) \\
-\frac{W_0(Q_0(p))}{f_0(Q_0(p))}
\end{array} \right] , \;\; p \in [p_1,p_2] \label{eq:limit_quantile}
\end{eqnarray}

The stochastic process $( \ref{eq:limit_quantile} )$ is a Gaussian process with zero mean function and covariance kernel:
\[C_z(p,u)=\begin{bmatrix}
\frac{C_1(Q_1(p),Q_1(u))}{f_1(Q_1(p)) f_1(Q_1(u))} & \frac{C_{10}(Q_1(p),Q_0(u))}{f_1(Q_1(p)) f_0(Q_0(u))} \\
\frac{C_{01}(Q_0(p),Q_1(u))}{f_0(Q_0(p)) f_1(Q_1(u))}  & \frac{C_0(Q_0(p),Q_0(u))}{f_0(Q_0(p)) f_0(Q_0(u))} \\
\end{bmatrix} . \]
Note that $Z ( \cdot ) \overset{d}{=}-\mathbb{Z} ( \cdot )$ due to the symmetry of the Gaussian distribution.

In \cite{firpo07} the difference between corresponding quantiles:
\begin{eqnarray}
\varphi ( p) = Q_1 (p) - Q_0 (p) \label{eq:qte}
\end{eqnarray}
\noindent is considered. It is known as {\em Quantile Treatment Effect} (QTE, for short). From $( \ref{eq:e36})$ it is
intuitive to estimate $\varphi ( p) $ by
\begin{eqnarray}
\widehat{\varphi} ( p) = \widehat{Q}_{1,n}(p)-\widehat{Q}_{0,n}(p)  \label{eq:qte_estim}
\end{eqnarray}
The estimator $( \ref{eq:qte_estim} )$ is asymptotically equivalent to the estimator of QTE defined in
\cite{firpo07}. In fact,
from $( \ref{eq:limit_quantile} )$ it appears that
\begin{eqnarray}
\sqrt{n} ( \widehat{\varphi} ( p) - \varphi ( p) ) =
\sqrt{n}(\widehat{Q}_{1,n}(p)-\widehat{Q}_{0,n}(p)-(Q_1(p)-Q_0(p)))
\label{eq:e37}
\end{eqnarray}
\noindent tends in distribution, as $n$ goes to infinity, to a Gaussian distribution with zero mean and variance:
\begin{eqnarray}
V & = & \frac{C_1(Q_1(p),Q_1(p))}{f_1(Q_1(p))^2}+\frac{C_0(Q_0(p),Q_0(p))}{f_0(Q_0(p))^2 }-\frac{C_{10}(Q_1(p),Q_0(p))}{f_1(Q_1(p))f_0(Q_0(p))}-\frac{C_{01}(Q_0(p),Q_1(p))}{f_0(Q_0(p))f_1(Q_1(p))} \nonumber \\
& = & \frac{C_1(Q_1(p),Q_1(p))}{f_1(Q_1(p))^2}+\frac{C_0(Q_0(p),Q_0(p))}{f_0(Q_0(p))^2 }-2\frac{C_{10}(Q_1(p),Q_0(p))}{f_1(Q_1(p))f_0(Q_0(p))} \nonumber \\
& = & \frac{1}{f_1(Q_1(p))^2}\biggl\{ E_x \biggr[\frac{1}{p(x)}F_1(Q_1(p) \vert x)(1-F_1(Q_1(p) \vert x))\biggr] \nonumber \\
& + & E_x \biggr[(F_1(Q_1(p) \vert x)-F_1(Q_1(p)))^2\biggr]\biggr\}
\nonumber \\
& + & \frac{1}{f_0(Q_0(p))^2}\biggl\{ E_x \biggr[\frac{1}{1-p(x)}F_0(Q_0(p) \vert x)(1-F_0(Q_0(p) \vert x))\biggr] \nonumber \\
& + & E_x \biggr[(F_0(Q_0(p) \vert x)-F_0(Q_0(p)))^2 \biggr]\biggr\} \nonumber \\
& - & 2\frac{1}{f_0(Q_0(p))f_1(Q_1(p))}E_x \biggr[(F_1(Q_1(p) \vert x)-F_1(Q_1(p))(F_0(Q_0(p) \vert x)-F_0(Q_0(p))\biggr] \nonumber \\
& = & \frac{1}{f_1(Q_1(p))^2}E_x \biggr[\frac{1}{p(x)}F_1(Q_1(p) \vert x)(1-F_1(Q_1(p) \vert x))\biggr] \nonumber \\
& + & \frac{1}{f_0(Q_0(p))^2}E_x \biggr[\frac{1}{1-p(x)}F_0(Q_0(p) \vert x)(1-F_0(Q_0(p) \vert x))\biggr] \nonumber \\
& + & E_x \biggr[\biggl(\frac{F_1(Q_1(p) \vert x)-p}{f_1(Q_1(p))}-\frac{F_0(Q_0(p) \vert x)-p}{f_0(Q_0(p))}\biggr)^2\biggr] .
\label{eq:e38}
\end{eqnarray}
\noindent which coincides with the asymptotic variance of the estimator of QTE used in \cite{firpo07}.

\section{Confidence bands for $F_1$ and $F_0$}
\label{sec:CB}

 The aim of the present section is to construct a confidence bandwidth for $F_1$, $F_0$, assuming again that they are continuous d.f.s..
As seen in Proposition $\ref{continuous-trajectories}$, under this assumption the process $W(\cdot)=[W_0(\cdot), \, W_1(\cdot)]^{\prime}$ has a.s. continuous trajectories. Furthermore:
\begin{eqnarray}
W_j(y)\overset{q.c.}{\rightarrow}0 \;; {\mathrm{as}} \; y \rightarrow \pm \infty , \; j=1, \, 0. \nonumber
\end{eqnarray}
In other words, the trajectories of $W(\cdot)$ are continuous and bounded with probability 1.
From now on, we will also assume that the cross-covariance matrix $C(y, \, t) = E\bigl [W(y) \otimes W(t) \bigr ]$ is such that $C(y, \,y)$
is a positive-definite matrix, for every real $y$.
Under these conditions it is possible to show (\cite{lifs82}) that the functional:
$ \sup_{y} \vert W_j(y) \vert $
can only have an atom at the point
\begin{eqnarray}
\sup_{y: \, V( W_{j} (y) ) = 0}  E  \left [ \left \vert W_j (y) \right \vert \right ] = 0 \nonumber
\end{eqnarray}
\noindent and has absolutely continuous distribution on $(0, \, + \infty )$.
On the other hand, $V(  W_{j} (y)  ) = 0$ only when $y \rightarrow \pm \infty$, and, from Th. 8.1 in \cite{dudley73} it follows that
 $\sup_{\vert y \vert \leq M} \vert W_j (y) \vert $ has absolutely continuous distribution in $(0, \, + \infty )$, for every positive $M$. Hence
\begin{eqnarray}
P \left ( \sup_{y \in \mathbb{R}} \left \vert W_j (y) \right \vert > 0 \right ) \geq
\lim_{M \rightarrow \infty} P \left ( \sup_{y \in \mathbb{R}} \left \vert W_j (y) \right \vert > 0 \right ) = 1 \nonumber
\end{eqnarray}
\noindent which proves that the distribution of $\sup_{y} \left \vert W_j (y) \right \vert $ has no atom at $0$. In other terms,
$\sup_{y} \left \vert W_j (y) \right \vert $  has absolutely continuous distribution on $(0, \, + \infty )$.

The starting point to construct a confidence band of asymptotic level $1-\alpha$  for $F_j (\cdot)$ consists in considering the Kolmogorov statistic:
\begin{eqnarray}
\underset{y}{\sup} \vert \widehat{F}_j(y)-F_j(y) \vert . \nonumber
\end{eqnarray}

From Propositions $\ref{weak-conv}$, $\ref{continuous-trajectories}$, we obtain
\begin{eqnarray}
\lim_{n \rightarrow \infty} P \left ( \sup_{y} \sqrt{n} \left \vert \widehat{F}_j(y)-F_j(y)
\right \vert \leq x
\right ) = P \left ( \sup_{y} \left \vert W_1(y) \right \vert \leq x \right ) \;\; \forall \, x \in \mathbb{R}
\label{eq:conv_kolm}.
\end{eqnarray}

Let $d_{j,1-\alpha}$ be the ${1-\alpha}$ quantile of the distribution of $\sup \vert W_j(y) \vert$. As a consequence of the absolute continuity
 of $ \sup \vert W_j (y) \vert$, there is a unique   $d_{j,1-\alpha}$ satisfying:
\begin{eqnarray}
P \left ( \sup_{y} \vert W_j(y) \vert \leq d_{j, 1-\alpha} \right ) =1-\alpha . \nonumber
\end{eqnarray}
The quantile $d_{j,1-\alpha}$ depends on unknown quantities. It can be estimated by subsampling.
Using the notation introduced in Section $\ref{subsampling_general}$, define
\begin{eqnarray}
\widehat{\theta}_{j,n} = \sqrt{n} \sup_{y} \left \vert \widehat{F}_{j,n}(y) - F_j(y) \right \vert , \;\;
\widehat{\theta}_{j,m} = \sqrt{m} \sup_{y} \left \vert \widehat{F}_{j,m}(y) - \widehat{F}_{j,n}(y) \right \vert . \nonumber
\end{eqnarray}
The subsampling procedure can be shortly described as follows.
\begin{enumerate}
\item Select $M$ independent subsamples of size $m$ from $\{A_i=(X_i, T_i, Y_i), \: i= 1, \, \dots , \, n \}$.
\item Compute  the  values:
$$\widehat{\theta}_{j,m;l}= \sqrt{m} \sup_{y} \left \vert \widehat{F}_{j,m;l}(y)-F_{j,n}(y) \right \vert , \; l=1, \, \dots , \, M$$.
\item Compute the empirical distribution function:
$$\widehat{R}_{j,n,m}(u)=\frac{1}{M} \sum_{l=1}^{M} I_{(\widehat{\theta}_{m,j;l}\le u )} . $$
\item Compute the quantile:
$$\widehat{d}_{j,1-\alpha}=R^{-1}_{j,n,m}(1-\alpha) =
\inf_{u} \left \{ R_{j,n,m}(u) \ge 1-\alpha \right \} .$$
\end{enumerate}

Now, it is easy to see that:
$$\widehat{R}_{j,n,m}(u) \overset{p}{\rightarrow} P \left ( \sup_{y} \left \vert W_j(y) \right \vert \leq u \right ) \;\;
{\mathrm{as}} \,
 n, \, m, \, M \rightarrow \infty , \; \frac{m}{n} \rightarrow 0.$$
From the absolute continuity of the distribution of $\sup_{y} \left \vert W_j (y) \right \vert$, it also follows that:
$$\widehat{d}_{j,1-\alpha}= \inf \left \{ y: \: \widehat{R}_{j,n,m}(y) \geq 1-\alpha \right \}$$
tends in probability to $d_{j,1-\alpha}$. In symbols:
$$\widehat{d}_{j,1-\alpha} \overset{d}{\rightarrow} d_{1-\alpha} \;\;
{\mathrm{as}} \; n, \, m, \,M \rightarrow \infty , \; \frac{m}{n} \rightarrow 0, \; \forall \, 0<\alpha<1 .$$

Finally, from $( \ref{eq:conv_kolm} )$ we may conclude that
\begin{eqnarray}
1-\alpha & \simeq & P \left ( \left \vert \widehat{F}_{j,n}(y)-F_j(y) \right \vert \leq
\frac{\widehat{d}_{j,1-\alpha}}{\sqrt{n}} \; \forall \, y \in \mathbb{R} \right ) \nonumber \\
& = & P \left ( \widehat{F}_{j,n}(y)-\frac{\widehat{d}_{j,1-\alpha}}{\sqrt{n}} \leq F_j(y) \leq
\widehat{F}_{j,n}(y)+\frac{\widehat{d}_{j,1-\alpha}}{\sqrt{n}} \; \forall \, y \in \mathbb{R} \right ) \nonumber \\
& = &
P \left ( \max \left ( 0, \, \widehat{F}_{j,n}(y) - \frac{\widehat{d}_{j,1-\alpha}}{\sqrt{n}} \right ) \leq F_j(y) \leq  \min \left (
1, \, \widehat{F}_{j,n}(y) + \frac{\widehat{d}_{j,1-\alpha}}{\sqrt{n}} \right ) \; \forall y \in \mathbb{R} \right ) \nonumber
\end{eqnarray}
\noindent so that the region
\begin{eqnarray}
\left \{ \left [ \max \left ( 0, \, \widehat{F}_{j,n}(y) - \frac{\widehat{d}_{j,1-\alpha}}{\sqrt{n}} \right ) , \: \min
\left ( 1, \, \widehat{F}_{j,n}(y) + \frac{\widehat{d}_{j,1-\alpha}}{\sqrt{n}} \right ) \right ] ;  y \in \mathbb{R} \right \} 
\label{eq:e39a}
\end{eqnarray}
\noindent is a confidence bandwidth for $F_j (\cdot)$ of asymptotic level $1-\alpha$.

\section{Testing for the presence of a treatment effect: two (sub)sample Wilcoxon test}
\label{sec:Wilcoxon}

\subsection{Wilcoxon type statistic}
\label{sub1_wilcox}

In nonparametric statistics, a problem of considerable relevance consists in testing for the possible difference between two samples. Among several
proposals, the two-sample Wilcoxon (or Wilcoxon-Mann-Whitney) test plays a central role in applications, mainly because of its properties.
The goal of the present section is to propose a Wilcoxon type statistic to test for the possible difference between the (sub)sample of treated subjects
and the (sub)sample of untreated subjects. In other terms, we aim at developing a Wilcoxon type statistic to test for the possible difference between treated and
untreated subjects, {\em i.e.} for the possible presence of a treatment effect.

From now on, we will assume $F_0$ and $F_1$ are both continuous. As in the classical Wilcoxon two-sample test, in order to measure the difference between the distributions of $Y_{(1)}$ and $Y_{(0)}$, we consider
\begin{eqnarray}
\theta_{01} = \theta ( F_0 , \, F_1 ) = \int_{\mathbb{R}} F_0(y) \, dF_1(y) .
\label{eq:e39}
\end{eqnarray}
\noindent The parameter $\theta_{01}$ $( \ref{eq:e39} )$ possesses a natural interpretation, because it is equal to the probability
that a treated subject possesses a $y$-value greater than the $y$-value for an independent, untreated subject. A few properties of
$\theta_{01}$ are listed below.
\begin{enumerate}
\item [1)] $\theta_{01}$ depends only on the marginal d.f.s $F_0$, $F_1$ (not on the way $Y_{(0)}$, $Y_{(1)}$ are associated in the same subject).
\item [2)] If $F_0 = F_1$ then $\theta_{01}=\frac{1}{2}$;
\item [3)] Using $\theta_{01}$ is equivalent to use $\theta_{10}=\int F_1(y) \, dF_0(y)$, as it it seen by an integration by parts.
\item [4)] If $F_1(y) \leq F_0(y)$ $\forall \, y \in \mathbb{R}$, {\em i.e.} if $Y_{(1)}$ is {\em stochastically larger} than $Y_{(0)}$, then:
$$\theta_{01}=1- \int_{\mathbb{R}} F_1(y) \, dF_0(y)  \geq 1-\int_{\mathbb{R}} F_0(y) \, dF_0(y) =\frac{1}{2} .$$
\end{enumerate}

The Wilcoxon type statistic we consider here is obtained in two steps, essentially by a plug-in approach.

\begin{enumerate}
\item [Step 1.] Estimation of the marginal d.f.s $F_1$, $F_0$:
\begin{eqnarray}
\widehat{F}_{j,n}(y)=\sum_{i=1}^{n} w^{(j)}_{i,n}I_{(Y_i \le y)}, \ w^{(j)}_{i,n}=\frac{I_{(T_i=1)}/\widehat{p}_{j,n}(x_i) }{\sum_{k=1}^{n} I_{(T_k=1)} /
\widehat{p}_{j,n} (x_k) } , \;\; j=1, \, 0.
\label{eq:e40}
\end{eqnarray}
\item [Step 2.] Estimation of $\theta_{01}$:
\begin{eqnarray}
\widehat{\theta}_{01,n} & = & \theta ( \widehat{F}_0 , \, \widehat{F}_1 ) \nonumber \\
& = & \int_{\mathbb{R}}
\widehat{F}_{0,n}(y) \, d \widehat{F}_{1,n}(y) \nonumber  \\
& = & \sum_{i=1}^{n} \sum_{k=1}^{n} w^{(1)}_{i,n} w^{(0)}_{k,n} I_{(y_{k} \leq y_{i})} .
\label{eq:e41}
\end{eqnarray}
\end{enumerate}
Note that
$w^{(1)}_{i,n} w^{(0)}_{k,n} \neq 0 $ if and only if (iff) $(I_{(T_i=1)}=1) \land (I_{(T_k=0)}=1)$, {\em i.e.} iff $i$ is treated and $k$ is untreated.
This essentially shows that $\widehat{\theta}_{01}$ is based on the comparison \textit{treated/untreated}.

The limiting distribution of the statistic $( \ref{eq:e41} )$ is obtained as a consequence of Proposition $\ref{weak-conv}$.

\begin{proposition}
\label{asymptotics_wilcoxon}
Assume that the conditions of Proposition $\ref{weak-conv}$ are fulfilled. Then
\begin{eqnarray}
\sqrt{n} ( \widehat{\theta}_{01,n}-\theta_{01}) \stackrel{d}{\rightarrow} N (0, \, V ) \;\; {\mathrm{as}} \; n \rightarrow \infty
\label{eq:e42}
\end{eqnarray}
\noindent where
\begin{eqnarray}
V =
E_x \left [ \frac{1}{p(x)} V \left ( F_0(Y_1) \vert x \right ) \right ]
+ E_x \left [ \frac{1}{1-p(x)} V \left ( F_1(Y_0) \vert x \right ) \right ]
+ V_x \left ( \gamma_{10}(x) - \gamma_{01}(x) \right )
\label{eq:e55}
\end{eqnarray}
\noindent and
\begin{eqnarray}
\gamma_{10}(x) = E[F_1(Y_0) \vert x] = \int_{\mathbb{R}} F_1(t \vert x) \, dF_0(t)  , \;\;
\gamma_{01}(x) = E[ F_0(Y_1) \vert x ] = \int_{\mathbb{R}}  F_0 (y \vert x) \, dF_1(y) .
\label{eq:def_gamma}
\end{eqnarray}
\end{proposition}
\noindent {\bf Proof.} See Appendix.

\subsection{Variance estimation}
\label{ssec:W_var_est}

The asymptotic variance $V$ appearing in $( \ref{eq:e55} )$ contains unknown terms, that can be consistently estimated on the basis of
sample data.
In particular, the estimation of $\gamma_{01}(x) = E [ I_{(T=1)} p(x)^{-1} F_0(Y) \vert x ]$ can be simply developed by considering the regression of
 $$\frac{I_{(T_i=1)}}{\widehat{p}_n(x_i)}\widehat{F}_{0n}(Y_i) , \;\; i = 1, \, \dots , \, n$$
  on $x_i$, $i = 1, \, \dots , \, n$, and to estimate the regression function by a method ensuring consistency ({\em e.g.} local polynomials, Nadaraya-Watson kernel regression, spline).
The resulting estimator $\widehat{\gamma}_{01,n}(x)$ is uniformly consistent on compact sets of $x$s under few regularity conditions.
In the same way, $\gamma_{10}(x)$ can be consistently estimated by $\widehat{\gamma}_{10,n}(x)$, say. As a consequence the term $V_x(\gamma_{10}(x)-\gamma_{01}(x))$ can be estimated by:
\begin{eqnarray}
\widehat{V}_{a,n} & = & \frac{1}{n}\sum_{i=1}^{n}
\left ( \widehat{\gamma}_{10,n}(x_i) - \widehat{\gamma}_{01,n}(x_i)
- \left ( \frac{1}{n}\sum_{i=1}^{n} ( \widehat{\gamma}_{10,n}(x_i) - \widehat{\gamma}_{01,n}(x_i) \right ) \right )^2 .
\label{eq:e56}
\end{eqnarray}
Note that as an alternative estimator, one could consider:
\begin{eqnarray}
\widehat{\widehat{V}}_{a,n}=\frac{1}{n}\sum_{i=1}^{n} \left ( \widehat{\gamma}_{10,n}(x_i)
- \widehat{\gamma}_{01,n}(x_i)\right )^2 - \left ( 1-\widehat{\theta}_{01,n} \right )^2 .
\nonumber
\end{eqnarray}
In the second place, we have to estimate
\begin{eqnarray}
E_x \left [ \frac{1}{p(x)} V(F_0(Y_1) \vert x) \right ]  =
E_x \left [ \frac{1}{p(x)} E[F_0(Y_1)^2 \vert x] \right ]
- E_x \left [ \frac{1}{p(x)} \gamma_{01}(x)^2 \right ] .
\label{eq:e58}
\end{eqnarray}
The term $E_x [ p(x)^{-1} \gamma_{01}(x)^2 ]$ can be estimated with
\begin{eqnarray}
\frac{1}{n} \sum_{i=1}^{n} \left [ \frac{1}{\widehat{p}_n(x_i)} \widehat{\gamma}_{01,n}(x_i)^2 \right ] .
\nonumber
\end{eqnarray}
The term:
\begin{eqnarray}
M_{01}(x) = E[F_0(Y_1)^2 \vert x]
= E \left [ \left . \frac{I_{(T=1)}}{p(x)}F_0(Y)^2 \right \vert x \right ] \nonumber
\end{eqnarray}
\noindent can be estimated by means of a non parametric regression of:
\begin{eqnarray}
\frac{I_{(T=1)}}{\widehat{p}_n(x_i)} \widehat{F}_{0,n}(Y_i)^2
\nonumber
\end{eqnarray}
\noindent
with respect to $x_i$s. The resulting estimator $\widehat{M}_{01,n}(x)$, say, is consistent under few conditions.
In the same way, an estimator $\widehat{M}_{10,n}(x)$ of
\begin{eqnarray}
M_{10}(x) = E[F_1(Y_0)^2 \vert x] = E \left [ \left . \frac{I_{(T=0)}}{1-p(x)} F_1(Y)^2 \right \vert x \right ] \nonumber
\end{eqnarray}
\noindent can be obtained.

The asymptotic variance of $\widehat{\theta}_{10,n}$ can be finally estimated by:
\begin{eqnarray}
\widehat{V}_n & = & \frac{1}{n}
\sum_{i=1}^{n}\frac{1}{\widehat{p}_{1,n} (x_i)} \left \{ \widehat{M}_{01,n}(x_i) - \widehat{\gamma}_{01,n}(x_i)^2 \right \} \nonumber \\
& + & \frac{1}{n}\sum_{i=1}^{n}\frac{1}{\widehat{p}_{0,n} (x_i)} \left \{
\widehat{M}_{10,n}(x_i) - \widehat{\gamma}_{10,n}(x_i)^2 \right \} + \widehat{V}_{a,n} .
\label{eq:e59}
\end{eqnarray}

\subsection{Testing the equality of $F_1$ and $F_0$ via Wilcoxon type statistic}
\label{ssec:W_subs}

A test for the equality of $F_1$ and $F_0$ can be constructed {\em via} the statistic $\widehat{\theta}_{01,n}$
$( \ref{eq:e41} )$. As already seen, when $F_1$ and $F_0$ coincide, $\theta_{01}$ is equal to $1/2$. Hence, the idea is to construct a test for
the hypotheses problem
\begin{eqnarray*}
\begin{cases}
H_0: \theta_{01}=\frac{1}{2}\\
H_1: \theta_{01}\ne\frac{1}{2}\\
\end{cases}
\end{eqnarray*}

On the basis of Proposition $\ref{asymptotics_wilcoxon}$, and the variance estimator $( \ref{eq:e59} )$, the region
\begin{eqnarray}
\sqrt{n} \left \vert \frac{\widehat{\theta}_{01,n}-\frac{1}{2}}{\sqrt{\widehat{V}_n}}
\right \vert \leq z_{\frac{\alpha}{2}} \label{eq:accept_wilcox_norm}
\end{eqnarray}
\noindent (where $z_{\frac{\alpha}{2}}$ is the $(1-\frac{\alpha}{2})$ quantile of the standard Normal distribution) is an acceptance region of asymptotic
significance level $\alpha$.

Alternatively, one could approximate the quantiles of the distribution of $\widehat{\theta}_{01,n}$ by subsampling, as outlined in Section
$ \ref{subsampling_general} $. Using the notation introduced for subsampling, it is seen that the acceptance region
\begin{eqnarray}
\left [ \widehat{\theta}_{01,n} - \frac{1}{\sqrt{n}} R^{-1}_{n,m} \left ( 1-\frac{\alpha}{2} \right )
, \: \widehat{\theta}_{01,n} - \frac{1}{\sqrt{n}} R^{-1}_{n,m} \left ( \frac{\alpha}{2} \right ) \right ]
\label{eq:conf_int_wilcox_subsamp}
\end{eqnarray}
\noindent is a confidence interval for $\theta_{01}$ of asymptotic level $1- \alpha$. Hence, the test consisting in rejecting $H_0$ whenever
the interval $( \ref{eq:conf_int_wilcox_subsamp} )$ does not contain $1/2$, possesses asymptotic significance level  $\alpha$.

\section{Testing for stochastic dominance}
\label{sec:dominance}

In evaluating the effect of a treatment, it is sometimes of interest to test wether the treatment itself has an effect on the {\em whole} distribution
function of $Y$, {\em i.e.} wether the treatment improves the behaviour of the whole d.f. of $Y$. Various forms of stochastic dominance are discussed in
\cite{mcfad:89}, \cite{anderson96}.
In particular, in the present section we will focus
on testing for first-order stochastic dominance.
The d.f. $F_1$ first-order stochastically dominates $F_0$ if $F_1(y) \leq F_0(y)$ $\forall \, y \in \mathbb{R}$. Our main goal is to construct a
test for the (uni-directional) hypotheses
\begin{eqnarray*}
\begin{cases}
H_0: \Delta(y) \le 0 \ \forall y \in \mathbb{R}\\
H_1: \Delta(y) > 0 \ \text{for  at  least one}  \ y >0\\
\end{cases}
\end{eqnarray*}
where $\Delta(y)=F_1(y)-F_0(y)$.

In econometrics and statistics, there is an extensive amount of literature on testing for stochastic dominance, since the
papers by  \cite{anderson96}, \cite{daviduclos00}. In \cite{linton05} a Kolmogorov-Smirnov type test is proposed, and a method
to construct critical values based on subsampling is proposed. For further bibliographic reference, and a deep analysis of contributions
to testing for stochastic dominance, cfr. the recent paper by \cite{donaldshu16}.

In the present paper, we confine ourselves to a simple, intuitive procedure to test for uni-directional dominance.
A simple idea to construct a test for the above hypotheses problem is to invert a confidence region for $\Delta ( \cdot )$. The null hypothesis $H_0$ is rejected whenever the confidence region has empty intersection with $H_0$. More formally, the test procedure we consider here is defined as follows.
\begin{itemize}
\item [(i)] Compute a confidence region for $\Delta(\cdot)$ of (at least asymptotic) level $1-\alpha$;
\item [(ii)] Reject $H_0$ if the confidence region for $\Delta(\cdot)$  and $H_0$ are disjoint,
that is if for at least a real $y$ the region has lower bound greater than zero.
\end{itemize}

From now on, we will assume that both $F_0$, $F_1$ are continuous d.f.s.
Using the arguments in Section $\ref{sec:CB}$, it is possible to see that the r.v.
\begin{eqnarray}
\sup_y \left ( W_1(y) - W_0(y) \right ) \nonumber
\end{eqnarray}
\noindent has absolutely continuous distribution, with $P \left (
\sup_y \left ( W_1(y)-W_0(y) \right ) \geq 0 \right ) =1$.
Hence, there exists  a single $d_{1-\alpha}$ such that
\begin{eqnarray}
P \left ( \sup_y \left ( W_1(y)-W_0(y) \right ) \leq d_{1-\alpha} \right ) = 1 - \alpha , \;\;  0 < \alpha < 1. \nonumber
\end{eqnarray}

The quantile $d_{1-\alpha}$ can be estimated by subsampling, as outlined in Section $\ref{subsampling_general}$.
Define
\begin{eqnarray}
\widehat{\Delta}_n (y) = \widehat{F}_{1,n} (y) - \widehat{F}_{0,n} (y) . \nonumber
\end{eqnarray}
A subsampling procedure to estimate $d_{1-\alpha}$ is described below.

\begin{enumerate}
\item Select $M$ independent subsamples of size $m$ from the sample of $(X_i,T_i,Y_i) $s, $i=1, \, \dots , \, n$.
 \item Compute the subsample statistics
\begin{eqnarray}
 \widehat{\theta}_{m,l}=\sqrt{m} \sup_y \left ( \widehat{F}_{1,m;l}(y)
 -\widehat{F}_{0,m;l}(y)  -  ( \widehat{F}_{1,n}(y) - \widehat{F}_{0,n}(y) \right ), \;\; l=1, \, \dots , \, M . \nonumber
\end{eqnarray}
 \item Compute the corresponding empirical d.f.
\begin{eqnarray}
\widehat{R}_{n,m}(u) = \frac{1}{M} \sum_{l=1}^{M} I_{( \widehat{\theta}_{m,l} \le u)} . \nonumber
\end{eqnarray}
 \item Compute the corresponding quantile
\begin{eqnarray}
\widehat{d}_{1-\alpha} = \widehat{R}_{n,m}^{-1}(u) = \inf \left \{ u: \; \widehat{R}_{n,m}(u) \ge 1-\alpha \right \} . \nonumber
\end{eqnarray}
\end{enumerate}

The arguments in Section $\ref{sec:CB}$ show that
\begin{eqnarray}
& \, & \widehat{R}_{n,m}(u) \overset{p}{\rightarrow} P \left ( \sup \left ( W_1(y) - W_0(y) \right ) \leq u \right ) \;\;
 \forall \, u \in \mathbb{R} , \;\;
{\mathrm{as}} \; n, \, m, \, M \rightarrow  \infty , \;  \frac{m}{n} \rightarrow 0 ; \nonumber \\
& \, & \widehat{d}_{1-\alpha} \overset{p}{\rightarrow} d_{1-\alpha} \;\;
 \forall \, 0 < \alpha < 1 , \;\;
{\mathrm{as}} \; n, \, m, \, M \rightarrow  \infty , \;  \frac{m}{n} \rightarrow 0 . \nonumber
\end{eqnarray}
Hence, the asymptotically exact approximation
\begin{eqnarray}
1-\alpha & \simeq & P \left ( \sup_y \left ( W_1(y)-W_0(y) \right ) \leq \widehat{d}_{1-\alpha} \right ) \nonumber \\
& \simeq & P \left ( \sup_y \sqrt{n} \left ( \widehat{\Delta}_n(y) - \Delta(y) \right ) \leq \widehat{d}_{1-\alpha} \right ) \nonumber \\
& = & P \left ( \Delta(y)  \geq \widehat{\Delta}_n(y) - \frac{\widehat{d}_{1-\alpha}}{\sqrt{n}} \; \forall \, y \in \mathbb{R} \right ) \nonumber
\end{eqnarray}
\noindent holds. As a consequence, the region
\begin{eqnarray}
\left \{ \left [ \widehat{\Delta}_n(y) - \frac{\widehat{d}_{1-\alpha}}{\sqrt{n}}; \: + \infty \right ) , \; y  \in \mathbb{R}  \right \} \nonumber
\end{eqnarray}
\noindent is a confidence region for $\Delta(\cdot)$ with asymptotic level $1-\alpha$. The null hypothesis $H_0$ is rejected whenever:
\begin{eqnarray}
\widehat{\Delta}_n (y) - \frac{\widehat{d}_{1- \alpha}}{\sqrt{n}} > 0 \; {\mathrm{for \; some}} \; y  \in \mathbb{R} .
\label{eq:stdom}
\end{eqnarray}
The performance of the testing procedure developed so far will be evaluated by simulation in Section $\ref{sec:simulation}$.

\section{A simulation study}
\label{sec:simulation}

The goal of the present section is to study by simulation the performance of the proposed methods for finite sample sizes.
In particular, estimation of $F_{j}$s and related hypotheses tests are studied under two scenarios: $(i)$ there is no treatment effect, {\em i.e.} $F_{1} $ coincides with $F_{0}$; $(ii)$ there is treatment effect, {\em i.e.} $F_{1} \neq F_0$.
\

$N=1000$ replications with samples sizes $n=1000$ and $n=5000$ have been generated by Monte Carlo simulation. The propensity score has been estimated {\em via}
the estimator considered in Th. $\ref{th_consist_sieve}$; the term $K$ has been chosen through least squares cross-validation.
As far as subsample approximation is concerned, $M=1000$ subsamples of size $m=100$ ($m=500$) have been drawn by simple random sampling from each of the $N=1000$ original samples of
size $n=1000$ ($5000)$.
\

In scenario $(i)$ (absence of treatment effect) the potential outcome $Y_{(j)}$ is specified as
\begin{eqnarray}
Y_{j}= 70 + 10  X + U_{j} , \;\; j=1, \, 0 \label{eq:noeffect}
\end{eqnarray}
\noindent where $X$ has a Bernoulli distribution with success probability $1/2$ ($X \sim Be (1/2))$ and $U_j$ has a uniform distribution
in the interval $[ - 10 , \, 10]$ ($U_j \sim U (-10 , \, 10)$). The r.v.s $U_1$, $U_0$ are mutually independent.
Clearly, $\theta_{01} = 1/2$, $E[Y_{(0)}]=E[Y_{(1)}]=75$, and $ATE=0$.

The exact distribution function of $Y_{(j)}$ is
\begin{eqnarray}
F_{j}(y) =
 \left \{\begin{array}{ll}
0& y<60 \\
 \frac{y-60}{40} \ (\frac{1}{2} \cdot \frac{y-60}{20})& 60 \leq y<70\\
 \frac{y-65}{20} \ (\frac{1}{2} \cdot \frac{y-60}{20}+\frac{1}{2} \cdot \frac{y-70}{20})& 70 \leq y<80 \\
\frac{y-50}{40}  \ (\frac{1}{2}+ \frac{1}{2} \cdot \frac{y-70}{20}) & 80 \leq y<90\\
1& y \geq 90 \\
\end {array} \right.
, \;\; j=1, \, 0. \label{eq:df_noeffect}
\end{eqnarray}
The d.f. $F_j$ $( \ref{eq:df_noeffect})$, and the corresponding density function $f_j$, are depicted in  Fig. $\ref{fig:Fvera}$.

\begin{figure}[htbp]
	\centering
	\includegraphics[height=3in, width=5in]{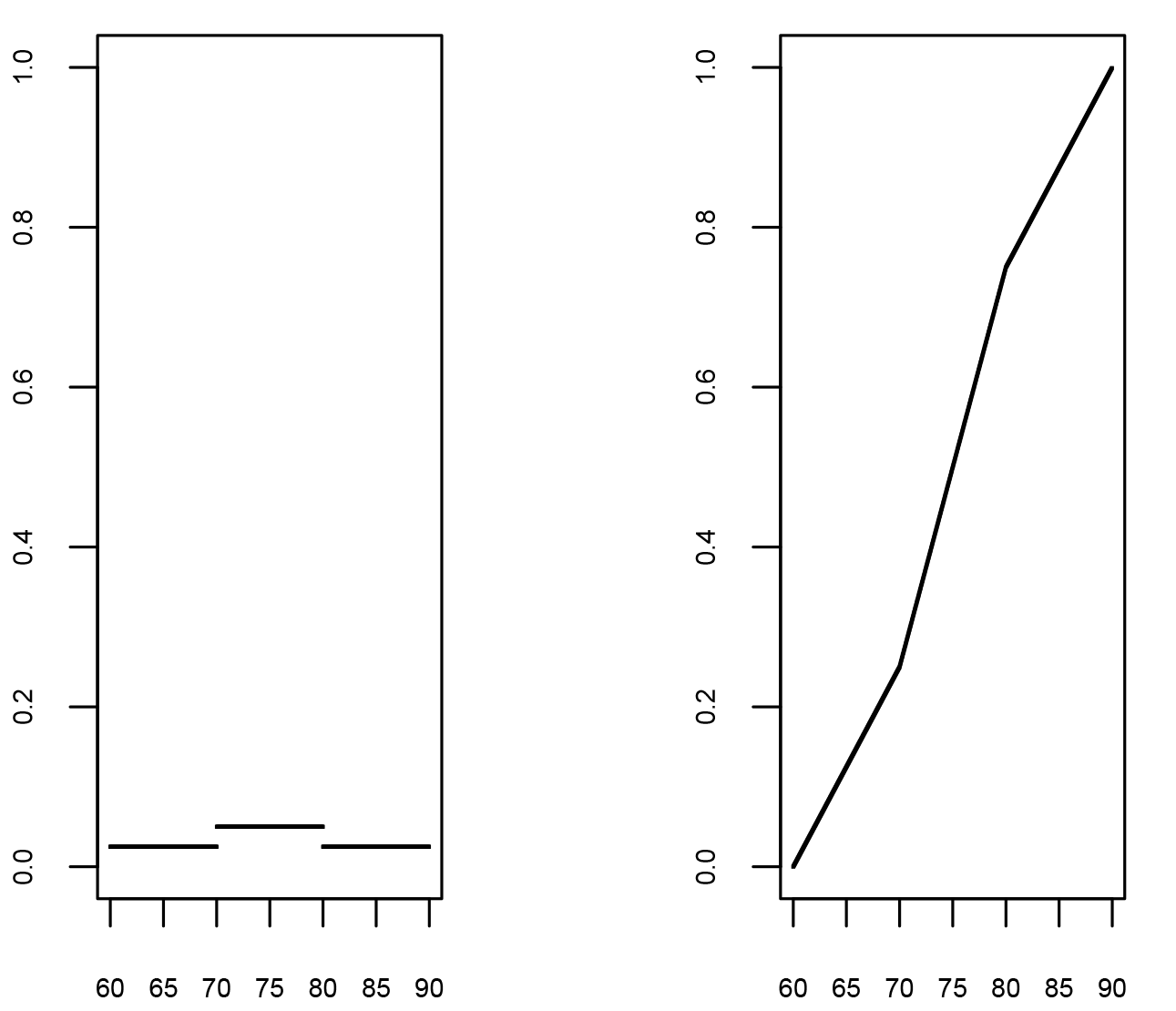}
	\caption{Density function and distribution function of $Y_{(0)}$, $Y_{(1)}$ under Scenario $(i)$}\label{fig:Fvera}
\end{figure}

The propensity score, in this case, is
 $$ p(x) = P(T=1 \vert X) = 0.75X+0.25(1-X)$$
Furthermore we have
$E[ Y \vert T=0]=72.5$ and $E[Y \vert T=1 ]=77.5$, so that $E[Y \vert T=1] - E[ Y \vert T=0]=5.0$ even if $ATE =0$. This is clearly
due to the confounding effect of $X$.

In Table \ref{tab:sim11} ($N=1000$, $n=1000$) and Table \ref{tab:sim12} ($N=1000$, $n=5000$) average, median and standard deviation of $\theta_{01}$, of $  \widehat{Q}_{j,n}(0.25)$, $ \widehat{Q}_{j,n}(0.50)$, and $ \widehat{Q}_{j,n}(0.75)$, $j=0,1$ are reported. The quantities are also reported for the estimator
$\widehat{\tau} = \sum_{i=1}^{n} y_i w^{(1)}_{i,n} - \sum_{i=1}^{n} y_i w^{(0)}_{i,n}$ of $ATE$ and for the ``naive'' mean difference between treated and untreated i.e.  $ n^{-1}\sum_{i=1}^{n} y_i - n^{-1}\sum_{i=1}^{n} y_i $.
\

In Tables \ref{tab:sim11a}-\ref{tab:sim11c} ($N=1000$, $n=1000$, $M=1000$, $m=100$) and Tables \ref{tab:sim12a}-\ref{tab:sim12c} ($N=1000$, $n=5000$, $M=1000$, $m=500$) the 95\% coverage probability and average length of confidence intervals for the Wilcoxon-type statistic $\hat{\theta}_{01,m}$ obtained via sampling and subsampling and for confidence bands for $F_1(y), F_0(y)$ and the percentage of rejection of the null hypothesis for the test of stochastic dominance are reported.

The results indicate that the Wilcoxon type statistic $\hat{\theta}_{01,m}$ and the estimated quantiles $ \widehat{Q}_{j,n}(p)$ perform well according to unbiasedness and dispersion. The sampling standard error of the Wilcoxon type statistic tends to be close to its theoretical one. The estimated ATE 
$\widehat{\tau} = \sum_{i=1}^{n} y_i w^{(1)}_{i,n} - \sum_{i=1}^{n} y_i w^{(0)}_{i,n}$ is equal to its ``true value'' (Tables \ref{tab:sim11} and \ref{tab:sim12}). The coverage probabilities of the confidence intervals are close to the nominal level 95\% (Tables \ref{tab:sim11a}-\ref{tab:sim11b} and \ref{tab:sim12a}-\ref{tab:sim12b}).  Finally, the percentage of rejection of the null hypothesis for the test of stochastic dominance is close to 0.05, being true the null hypothesis of no treatment effect in scenario $(i)$ i.e. $F_1=F_0$ (Tables \ref{tab:sim11c} and \ref{tab:sim12c}).

\begin{table}
\centering
\begin{tabular}{ccccc}
        \hline
        \hline
       \textit{ \textbf{Estimator}} & \textit{\textbf{Parameter value}} & \textit{\textbf{Average}} & \textit{\textbf{Median}} & \textit{\textbf{Standard Deviation}} \\
        \hline
          $\hat{\theta}_{01,m}$&0.50&0.50 &0.50 & 0.015\\
          $\widehat{Q}_{1,n}(0.25)$&70 &70.11 &70.20 &0.578    \\
       $\widehat{Q}_{1,n}(0.50)$&75 &75.37 &75.42 &0.318    \\
         $\widehat{Q}_{1,n}(0.75)$&80 &80.11 &80.15 & 0.158   \\
         $\widehat{Q}_{0,n}(0.25)$&70 &70.17 &70.16 &0.130     \\
       $\widehat{Q}_{0,n}(0.50)$&75 & 75.28&75.31 &0.307    \\
         $\widehat{Q}_{0,n}(0.75)$&80 &80.15 &80.03 &0.514    \\
           $\sum_{i=1}^{n} y_i w^{(1)}_{i,n} - \sum_{i=1}^{n} y_i w^{(0)}_{i,n} $&0&0.07&&0.411\\
           $ n^{-1}\sum_{i=1}^{n} y_i - n^{-1}\sum_{i=1}^{n} y_i $&0&5.04&&0.460\\
           \hline
      \end{tabular}
  \caption{Scenario $(i)$ - $N=1000$, $n=1000$, $\theta_{01}=0.50$ - sampling simulation results} \label{tab:sim11}
  \end{table}
  \begin{table}
 \vspace{3 mm}
 
\centering
\begin{tabular}{ccc}
        \hline
        \hline
   Parameter& \textit{\textbf{Coverage probability}} &\textit{ \textbf{Average length}}\\
        \hline
                    ${\theta}_{01,n}$ &0.95   &0.063 \\ 
          ${\theta}_{01,m}$ &0.94   &0.062   \\
        \hline
      \end{tabular}
  \caption{Scenario $(i)$ - $N=1000$, $n=1000$, $M=1000$, $m=100$, $\theta_{01}=0.500$} - Coverage probability and average length of confidence intervals for $\theta_{01}$ based on normal approximation and on subsampling (nominal level 0.95)
  \label{tab:sim11a}
  \end{table}
    \begin{table}
\centering
\begin{tabular}{ccc}
        \hline
        \hline
   \textit{\textbf{Parameter}}& \textit{\textbf{Coverage probability}} &\textit{ \textbf{Average length}}\\
        \hline
        $\underset{y}{sup}|{F}_{1,n}(y)-F_{1,n}(y)| $ &   0.98&0.137  \\
         $\underset{y}{sup}|{F}_{0,n}(y)-F_{0,n}(y)| $&   0.98&0.138  \\
        \hline
      \end{tabular}
  \caption{Scenario $(i)$ - $N=1000$, $n=1000$, $M=1000$, $m=100$, $\theta_{01}=0.50$} - Coverage probability and average length of confidence bands for $F_1, F_0$ (nominal level 0.95)
  \label{tab:sim11b}
  \end{table}
    \begin{table}
\centering
\begin{tabular}{cc}
        \hline
        \hline
\textit{\textbf{Test statistic}}&\textit{ \textbf{Rejection probability}}\\
        \hline
        $\Delta(y)$    &0.06  \\
        \hline
      \end{tabular}
  \caption{Scenario $(i)$ - $N=1000$, $n=1000$, $M=1000$, $m=100$, $\theta_{01}=0.50$} - Rejection probability of stochastic dominance test (nominal level 0.95)
  \label{tab:sim11c}
  \end{table}
\

\begin{table}
\centering
\begin{tabular}{ccccc}
        \hline
        \hline
       \textit{ \textbf{Estimator}} & \textit{\textbf{Parameter value}} & \textit{\textbf{Average}} & \textit{\textbf{Median}} & \textit{\textbf{Standard Deviation}} \\
        \hline
          $\hat{\theta}_{01,m}$&0.50&0.50 &0.50 & 0.007\\
          $\widehat{Q}_{1,n}(0.25)$&70 &69.92 &69.96&0.319    \\
       $\widehat{Q}_{1,n}(0.50)$&75 &74.98 &74.98 &0.168    \\
         $\widehat{Q}_{1,n}(0.75)$&80 &79.74 &79.75 & 0.069   \\
         $\widehat{Q}_{0,n}(0.25)$&70 &69.95 &69.97 &0.111     \\
       $\widehat{Q}_{0,n}(0.50)$&75 & 74.99&74.98 &0.146    \\
         $\widehat{Q}_{0,n}(0.75)$&80 &79.75 &79.77 &0.209    \\
           $\sum_{i=1}^{n} y_i w^{(1)}_{i,n} - \sum_{i=1}^{n} y_i w^{(0)}_{i,n} $&0&0.00&&0.184\\
           $ n^{-1}\sum_{i=1}^{n} y_i - n^{-1}\sum_{i=1}^{n} y_i $&0&4.97&&0.192\\
           \hline
      \end{tabular}
  \caption{Scenario $(i)$ - $N=1000$, $n=5000$, $\theta_{01}=0.50$ - sampling simulation results} \label{tab:sim12}
  \end{table}
  \begin{table}
\centering
\begin{tabular}{ccc}
        \hline
        \hline
   \textit{\textbf{Parameter}}& \textit{\textbf{Coverage probability}} &\textit{ \textbf{Average length}}\\
        \hline
                    ${\theta}_{01,n}$ &0.96   &0.028 \\ 
          ${\theta}_{01,m}$ &0.95   &0.027  \\
        \hline
      \end{tabular}
  \caption{Scenario $(i)$ - $N=1000$, $n=5000$, $M=1000$, $m=500$, $\theta_{01}=0.500$} - Coverage probability and average length of confidence intervals for $\theta_{01}$ based on normal approximation and on subsampling (nominal level 0.95)
  \label{tab:sim12a}
  \end{table}
  
\begin{table}
\centering
\begin{tabular}{ccc}
        \hline
        \hline
   \textit{\textbf{Parameter}}& \textit{\textbf{Coverage probability}} &\textit{ \textbf{Average length}}\\
        \hline
        $\underset{y}{sup}|{F}_{1,n}(y)-F_{1,n}(y)| $&   0.96&0.060  \\
         $\underset{y}{sup}|{F}_{0,n}(y)-F_{0,n}(y)| $ &   0.96&0.061  \\
        \hline
      \end{tabular}
  \caption{Scenario $(i)$ - $N=1000$, $n=5000$, $M=1000$, $m=500$, $\theta_{01}=0.50$} - Coverage probability and average length of confidence bands for $F_1, F_0$ (nominal level 0.95)
  \label{tab:sim12b}
  \end{table}

    \begin{table}
\centering
\begin{tabular}{cc}
        \hline
        \hline
\textit{\textbf{Test statistic}}&\textit{ \textbf{Rejection probability}}\\
        \hline
        $\Delta(y)$   &0.05 \\
        \hline
      \end{tabular}
  \caption{Scenario $(i)$ - $N=1000$, $n=5000$, $M=1000$, $m=500$, $\theta_{01}=0.50$} - Rejection probability of stochastic dominance test (nominal level 0.95)
  \label{tab:sim12c}
  \end{table}
\

In scenario $(ii)$ (presence of treatment effect), the potential outcome $Y_{(0)}$ is specified as in $( \ref{eq:noeffect} )$ with $j=0$. The potential outcome $Y_{(1)}$ is specified as
\begin{eqnarray}
Y_{(1)} = 75 + 10 \cdot X + U_{1} \label{eq:treat_effect}
\end{eqnarray}
\noindent  where $X$ has a Bernoulli distribution $X \sim Be(0.5)$ and $U_{0}$, $U_{1}$ have a Uniform distribution $U_1 \sim U[-10;10]$.
The r.v.s $X$, $U_{0}$, $U_{1}$ are mutually independent.

The exact distribution function of $Y_{(1)}$ is reported below
\begin{eqnarray}
F_{1}(y)=\left \{\begin{array}{ll}
0& y<65 \\
 \frac{y-65}{40} & 65 \leq y<75\\
 \frac{y-70}{20}  & 75 \leq y<85 \\
\frac{y-50}{40}    & 85 \leq y<95\\
1& y \geq 95 \\
\end {array} \right .  \label{eq:df_yeseffect}
\end{eqnarray}
\noindent and depicted in Fig. \ref{fig:Fvere06}.
\begin{figure}[htbp]
	\centering
	\includegraphics[height=5in, width=6in]{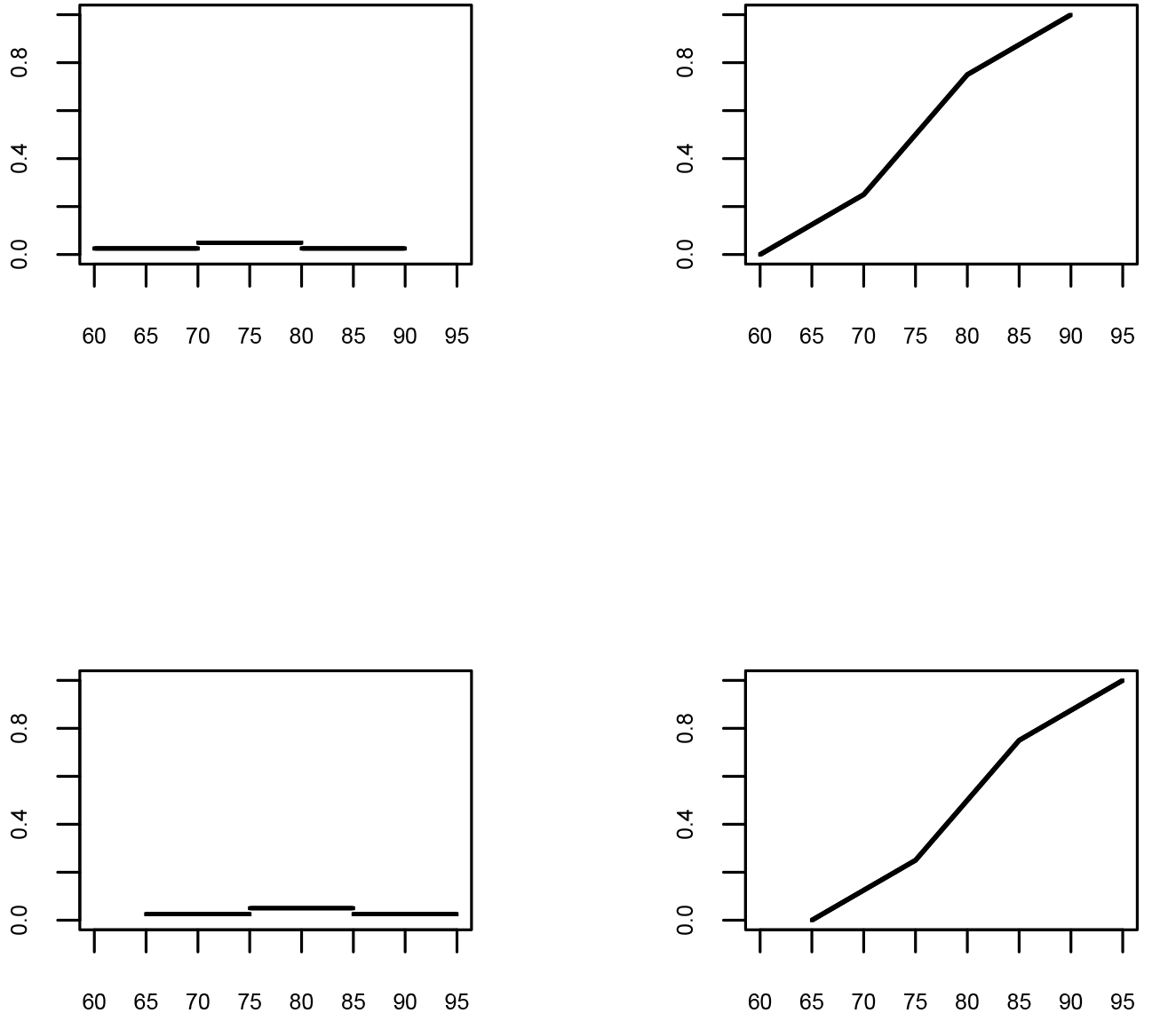}
	\caption{Density function and distribution function of $Y_{(0)}$ (top), $Y_{(1)} $ (bottom)
under Scenario $(ii)$} \label{fig:Fvere06}
\end{figure}

In scenario $(ii)$, we have  $\theta_{01} = 0.67$,
$E[Y_{(0)}]=75$, $E[Y_{(1)}]=80$, and then $ATE=5$. Furthermore, $F_{1}$ stochastically dominates $F_0$.

The propensity score is 
$$p(X) = P(T=1 \vert X)=0.25X+0.75(1-X)$$

\noindent so that $E[Y \vert T=0]=77.5$ and $E[ Y \vert T=1 ]=77.5$ even if $ATE \neq 0$.
As in scenario $(i)$, this is due to the confounding effect of $X$.

\

In Table \ref{tab:sim21} ($N=1000$, $n=1000$) and Table \ref{tab:sim22} ($N=1000$, $n=5000$) average, median and standard deviation of $\theta_{01}$, of $  \widehat{Q}_{j,n}(0.25)$, $ \widehat{Q}_{j,n}(0.50)$, and $ \widehat{Q}_{j,n}(0.75)$, $j=0,1$ are reported. The quantities are also reported for the estimator
$\widehat{\tau} = \sum_{i=1}^{n} y_i w^{(1)}_{i,n} - \sum_{i=1}^{n} y_i w^{(0)}_{i,n}$ of $ATE$ and for the ``naive'' mean difference between treated and untreated i.e.  $ n^{-1}\sum_{i=1}^{n} y_i - n^{-1}\sum_{i=1}^{n} y_i $.
\

In Tables \ref{tab:sim21a}-\ref{tab:sim21c} ($N=1000$, $n=1000$, $M=1000$, $m=100$) and Tables \ref{tab:sim22a}-\ref{tab:sim22c} ($N=1000$, $n=5000$, $M=1000$, $m=500$) the 95\% coverage probability and average length of confidence intervals for the Wilcoxon-type statistic $\hat{\theta}_{01,m}$ obtained via sampling and subsampling and for confidence bands for $F_1(y), F_0(y)$ and the percentage of rejection of the null hypothesis for the test of stochastic dominance are reported.

The results indicate that the Wilcoxon type statistic $\hat{\theta}_{01,m}$ and the estimated quantiles $ \widehat{Q}_{j,n}(p)$ perform well according to unbiasedness and dispersion. The sampling standard error of the Wilcoxon type statistic tends to be close to its theoretical one. The estimated ATE 
$\widehat{\tau} = \sum_{i=1}^{n} y_i w^{(1)}_{i,n} - \sum_{i=1}^{n} y_i w^{(0)}_{i,n}$ is equal to its ``true value'' (Tables \ref{tab:sim21} and \ref{tab:sim22}). The coverage probabilities of the confidence intervals are close to the nominal level 95\% (Tables \ref{tab:sim21a}-\ref{tab:sim21b} and \ref{tab:sim22a}-\ref{tab:sim22b}).  Finally, the percentage of rejection of the null hypothesis for the test of stochastic dominance is close to 0.05, being true the null hypothesis of no treatment effect. As in scenario $(ii)$ i.e. $F_1$ stochastically dominates $F_0$ the rejection probability is smaller than in  in scenario $(i)$ (Tables \ref{tab:sim21c} and \ref{tab:sim22c}).
\begin{table}
\centering
\begin{tabular}{ccccc}
        \hline
        \hline
       \textit{ \textbf{Estimator}} & \textit{\textbf{Parameter value}} & \textit{\textbf{Average}} & \textit{\textbf{Median}} & \textit{\textbf{Standard Deviation}} \\
        \hline
          $\hat{\theta}_{01,m}$&0.67&0.67 &0.67 & 0.012\\
          $\widehat{Q}_{1,n}(0.25)$&75 &74.70 &74.69&0.289    \\
       $\widehat{Q}_{1,n}(0.50)$&80 &79.73 &79.75 &0.350    \\
         $\widehat{Q}_{1,n}(0.75)$&85 &85.17 &84.92 & 0.669   \\
         $\widehat{Q}_{0,n}(0.25)$&70 &69.64 &70.03 &0.698     \\
       $\widehat{Q}_{0,n}(0.50)$&75 & 74.81&74.90 &0.355    \\
         $\widehat{Q}_{0,n}(0.75)$&80 &79.90 &79.77 &0.352    \\
           $\sum_{i=1}^{n} y_i w^{(1)}_{i,n} - \sum_{i=1}^{n} y_i w^{(0)}_{i,n} $&5&4.98&&0.415\\
           $ n^{-1}\sum_{i=1}^{n} y_i - n^{-1}\sum_{i=1}^{n} y_i $&5&-0.03&&0.033\\
           \hline
      \end{tabular}
  \caption{Scenario $(i)$ - $N=1000$, $n=1000$, $\theta_{01}=0.672$ - sampling simulation results} \label{tab:sim21}
  \end{table}
  
  \begin{table}
\centering
\begin{tabular}{ccc}
        \hline
        \hline
   \textit{\textbf{Parameter}}& \textit{\textbf{Coverage probability}} &\textit{ \textbf{Average length}}\\
        \hline
                    ${\theta}_{01,n}$ &0.97   &0.051 \\ 
          ${\theta}_{01,m}$ &0.96   &0.049  \\
        \hline
      \end{tabular}
  \caption{Scenario $(i)$ - $N=1000$, $n=1000$, $M=1000$, $m=100$, $\theta_{01}=0.672$} - Coverage probability and average length of confidence intervals for $\theta_{01}$ based on normal approximation and on subsampling (nominal level 0.95)
  \label{tab:sim21a}
  \end{table} 
\begin{table}
\centering
\begin{tabular}{ccc}
        \hline
        \hline
   \textit{\textbf{Parameter}}& \textit{\textbf{Coverage probability}} &\textit{ \textbf{Average length}}\\
        \hline
        $\underset{y}{sup}|{F}_{1,n}(y)-F_{1,n}(y)| $ &   0.97&0.138  \\
         $\underset{y}{sup}|{F}_{0,n}(y)-F_{0,n}(y)| $ &   0.96&0.137  \\
        \hline
      \end{tabular}
  \caption{Scenario $(i)$ - $N=1000$, $n=1000$, $M=1000$, $m=100$, $\theta_{01}=0.672$} - Coverage probability and average length of confidence bands for $F_1, F_0$ (nominal level 0.95)
  \label{tab:sim21b}
  \end{table}

    \begin{table}
\centering
\begin{tabular}{cc}
        \hline
        \hline
\textit{\textbf{Test statistic}}&\textit{ \textbf{Rejection probability}}\\
        \hline
        $\Delta(y)$    &0.00  \\
        \hline
      \end{tabular}
  \caption{Scenario $(i)$ - $N=1000$, $n=1000$, $M=1000$, $m=100$, $\theta_{01}=0.672$} - Rejection probability of stochastic dominance test (nominal level 0.95)
  \label{tab:sim21c}
  \end{table}

\begin{table}
\centering
\begin{tabular}{ccccc}
        \hline
        \hline
       \textit{ \textbf{Estimator}} & \textit{\textbf{Parameter value}} & \textit{\textbf{Average}} & \textit{\textbf{Median}} & \textit{\textbf{Standard Deviation}} \\
        \hline
          $\hat{\theta}_{01,m}$&0.67&0.67 &0.67 & 0.005\\
          $\widehat{Q}_{1,n}(0.25)$&75 &74.85 &74.88&0.145    \\
       $\widehat{Q}_{1,n}(0.50)$&80 &79.72 &79.76 &0.137    \\
         $\widehat{Q}_{1,n}(0.75)$&85 &84.89 &84.84 & 0.243   \\
         $\widehat{Q}_{0,n}(0.25)$&70 &69.91 &70.00 &0.257     \\
       $\widehat{Q}_{0,n}(0.50)$&75 & 74.94&74.97 &0.160    \\
         $\widehat{Q}_{0,n}(0.75)$&80 &79.76 &79.76&0.090    \\
           $\sum_{i=1}^{n} y_i w^{(1)}_{i,n} - \sum_{i=1}^{n} y_i w^{(0)}_{i,n} $&5&4.98&&0.174\\
           $ n^{-1}\sum_{i=1}^{n} y_i - n^{-1}\sum_{i=1}^{n} y_i $&5&-0.04&&0.430\\
           \hline
      \end{tabular}
  \caption{Scenario $(i)$ - $N=1000$, $n=5000$, $\theta_{01}=0.67$ - sampling simulation results} \label{tab:sim22}
  \end{table}
  \begin{table}
\centering
\begin{tabular}{ccc}
        \hline
        \hline
   Parameter& \textit{\textbf{Coverage probability}} &\textit{ \textbf{Average length}}\\
        \hline
                    ${\theta}_{01,n}$ &0.96   &0.023 \\ 
          ${\theta}_{01,m}$ &0.95   &0.022  \\
        \hline
      \end{tabular}
  \caption{Scenario $(i)$ - $N=1000$, $n=5000$, $M=1000$, $m=500$, $\theta_{01}=0.67$} - Coverage probability and average length of confidence intervals for $\theta_{01}$ based on normal approximation and on subsampling (nominal level 0.95)
  \label{tab:sim22a}
  \end{table}
  
\begin{table}
\centering
\begin{tabular}{ccc}
        \hline
        \hline
   Parameter& \textit{\textbf{Coverage probability}} &\textit{ \textbf{Average length}}\\
        \hline
        $\underset{y}{sup}|{F}_{1,n}(y)-F_{1,n}(y)| $&   0.97&0.060  \\
         $\underset{y}{sup}|{F}_{0,n}(y)-F_{0,n}(y)| $ &   0.97&0.061  \\
        \hline
      \end{tabular}
  \caption{Scenario $(i)$ - $N=1000$, $n=5000$, $M=1000$, $m=500$, $\theta_{01}=0.50$} - Coverage probability and average length of confidence bands for $F_1, F_0$ (nominal level 0.95)
  \label{tab:sim22b}
  \end{table}

    \begin{table}
\centering
\begin{tabular}{cc}
        \hline
        \hline
Test statistic&\textit{ \textbf{Rejection probability}}\\
        \hline
        $\Delta(y)$    &0.00  \\
        \hline
      \end{tabular}
  \caption{Scenario $(i)$ - $N=1000$, $n=5000$, $M=1000$, $m=500$, $\theta_{01}=0.50$} - Rejection probability of stochastic dominance test (nominal level 0.95)
  \label{tab:sim22c}
  \end{table}

\newpage

{\Large \noindent {\bf Appendix - Technical Lemmas and proofs}}

\begin{lemma}
\label{lemma1}
Under the assumptions of Th. $\ref{th_consist_sieve}$:
\begin{eqnarray}
\sup_{x} \left \vert \frac{1}{\widehat{p}_{j,n}(x)} - \frac{1}{{p}_{j} (x)} \right \vert \overset{p}{\rightarrow}0 \; \; {\mathrm{as}}  \; n \rightarrow \infty, \;\; j=1, \, 0.
\label{eq:e3}
\end{eqnarray}
\end{lemma}
\noindent
\begin{proof}[{\bf Proof of Lemma \ref{lemma1}}]
 Take an arbitrary $0 < \epsilon <1$. Since $p_{1} (x) = p(x)$, $\widehat{p}_{1,n}(x) = \widehat{p}_{n}(x)$, we may write
\begin{eqnarray}
& \, & P \left ( \sup_x \left \vert   \frac{1}{\widehat{p}_{1,n}(x)} - \frac{1}{p_{1}(x)} \right \vert > \epsilon \right )
=  P \left ( \sup_x \frac{\left \vert \widehat{p}_{n}(x) - p(x) \right \vert}{\widehat{p}_{n}(x)} > \delta \epsilon \right ) \nonumber \\
\, & \leq &
P \left ( \sup_x \frac{\left \vert \widehat{p}_{n}(x) - p(x) \right \vert}{\widehat{p}_{n}(x)} > \delta \epsilon , \;
\sup_x \left \vert \widehat{p}_{n}(x) - p(x) \right \vert \leq \delta \epsilon \right ) +
P \left ( \sup_x \left \vert \widehat{p}_{n}(x) - p(x) \right \vert > \delta \epsilon \right ) \nonumber \\
\, & \leq &
P \left ( \sup_x \frac{\left \vert \widehat{p}_{n}(x) - p(x) \right \vert}{p(x) - \delta \epsilon} > \delta \epsilon \right ) +
P \left ( \sup_x \left \vert \widehat{p}_{n}(x) - p(x) \right \vert > \delta \epsilon \right ) \nonumber \\
\, & \leq &
P \left ( \sup_x \frac{\left \vert \widehat{p}_{n}(x) - p(x) \right \vert}{\delta (1- \epsilon )} > \delta \epsilon \right ) +
P \left ( \sup_x \left \vert \widehat{p}_{n}(x) - p(x) \right \vert > \delta \epsilon \right ) \nonumber \\
\, & \leq &
2 P \left ( \sup_x \left \vert \widehat{p}_{n}(x) - p(x) \right \vert > \delta^2 \epsilon (1- \epsilon ) \right )
\rightarrow 0 \;\; {\mathrm{as}} \; n \rightarrow \infty .
\label{eq:e4}
\end{eqnarray}
Since $( \ref{eq:e4} )$ holds for every positive $\epsilon$ small enough, the lemma is proved as $j=1$. The case $j=0$ is similar.
\end{proof}

\begin{lemma}
\label{lemma2}
Under the assumptions of Th. $\ref{th_consist_sieve}$:
\begin{eqnarray}
\frac{1}{n} \sum_{i=1}^{n} \frac{I_{(T_i=j)}}{\widehat{p}_{j,n} (x_i) } \overset{p}{\rightarrow} 1 , \;\;\; {\mathrm{as}} \; n \rightarrow \infty ,
\; j=1, \, 0.
\label{eq:e7}
\end{eqnarray}
\end{lemma}
\begin{proof}[{\bf Proof of Lemma \ref{lemma2}}]
Consider the case $j=1$. First of all, we have
\begin{eqnarray}
\frac{1}{n} \sum_{i=1}^{n} \frac{I_{(T_i=1)}}{\widehat{p}_{1,n} (x_i) } =
\frac{1}{n} \sum_{i=1}^{n} \left ( \frac{1}{\widehat{p}_{n} (x_i)}-\frac{1}{{p}(x_i)} \right )
I_{(T_i=1)} + \frac{1}{n} \sum_{i=1}^{n} \frac{I_{(T_i=1)}}{p(x_i) }  .
\label{eq:e8}
\end{eqnarray}
\noindent
Next, by Lemma $\ref{lemma1}$ it is easy to see that
\begin{eqnarray}
\left \vert \frac{1}{n} \sum_{i=1}^{n} \left ( \frac{1}{\widehat{p}_{n} (x_i)} - \frac{1}{{p}(x_i)} \right ) I_{(T_i=1)}\right \vert
& \leq & \underset{1 \le i \le n}{\max} \left \vert \frac{1}{\widehat{p}_{n} (x_i)}-\frac{1}{{p}(x_i)} \right \vert
\frac{1}{n} \sum_{i=1}^{n}I_{(T_i=1)} \nonumber \\
\, & \leq & \underset{x}{\sup} \left \vert \frac{1}{\widehat{p}_{n} (x)}-\frac{1}{{p}(x)} \right \vert  \overset{p}{\rightarrow}0
\;\; {\mathrm{as}} \; n \rightarrow \infty .
\label{eq:e9}
\end{eqnarray}
Furthermore, from the Strong Law of Large Numbers for sequences of {\em i.i.d.} r.v.s it is seen that
\begin{eqnarray}
\frac{1}{n} \sum_{i=1}^{n} \frac{I_{(T_i=1)}}{p(x_i)}\overset{a.s.}{\rightarrow} E \left [ \frac{I_{(T_i=1)}}{p(x_i) } \right ] = E_x \left [
 \frac{1}{p(x_i)} E \left [ \left . I_{(T_i=1)} \right \vert x_i \right ] \right ] = 1 \label{eq:e10}
\end{eqnarray}
\noindent as $n \rightarrow \infty$. From $( \ref{eq:e9} )$ and $( \ref{eq:e10} )$ the first convergence in $( \ref{eq:e8} )$ follows. Convergence in
the case $j=0$ is proved in a similar way.
\end{proof}

\begin{lemma}
\label{lemma3}
Consider the  ``pseudo-estimator'' of $F_{j}(y)$:
\begin{eqnarray}
\widetilde{F}_{j,n}(y) = \frac{1}{n} \sum_{i=1}^{n} \frac{I_{(T_i=j)}}{p_{j} (x_i) }I_{(Y_i \le y)} , \;\; j=1, \, 0.
\label{eq:e11}
\end{eqnarray}
Under the assumptions of Th. $\ref{th_consist_sieve}$:
\begin{eqnarray}
\underset{y}{\sup} \left \vert \widehat{F}_{j,n}(y)- \widetilde{F}_{j,n}(y) \right \vert
\overset{p} {\rightarrow} 0 \;\; {\mathrm{as}} \; n \rightarrow \infty , \; j=1, \, 0.
\label{eq:lem_eqiv}
\end{eqnarray}
\end{lemma}
\begin{proof}[{\bf Proof of Lemma \ref{lemma3}}]
Consider first the case $j=1$. From
\begin{eqnarray}
\widehat{F}_{1,n}(y)  - \widetilde{F}_{j,n}(y) & = &
\left ( \frac{1}{n} \sum_{i=1}^{n} \frac{I_{(T_i=1)}}{\widehat{p}_{1,n} (x_i)} \right )^{-1}
\left ( \frac{1}{n} \sum_{i=1}^{n} \left (  \frac{1}{\widehat{p}_{1,n} (x_i)} - \frac{1}{p_{1} (x_i)} \right )
I_{(T_i=1)} I_{(Y_i \le y)} \right ) \nonumber \\
& + &
\left \{ \left ( \frac{1}{n} \sum_{i=1}^{n} \frac{I_{(T_i=1)}}{\widehat{p}_{1,n} (x_i)} \right )^{-1} -1 \right \}
\left ( \frac{1}{n} \sum_{i=1}^{n} \frac{I_{(T_i=1)}}{p_{1} (x_i)} I_{(Y_i \le y)} \right ) \nonumber
\end{eqnarray}
\noindent it is seen that
\begin{eqnarray}
\underset{y}{\sup} \left \vert \widehat{F}_{1,n}(y)  - \widetilde{F}_{1}(y)) \right \vert & \leq &
\sup_{x} \left \vert \frac{1}{\widehat{p}_{n}(x)} - \frac{1}{p(x)} \right \vert \nonumber \\
& + &
\left \vert \left ( \frac{1}{n} \sum_{i=1}^{n} \frac{I_{(T_i=1)}}{\widehat{p}_n(x_i)} \right )^{-1} - 1 \right \vert
\left ( \frac{1}{n} \sum_{i=1}^{n} \frac{I_{(T_i=1)}}{p(x_i)} \right ) . \label{eq:interm_lemma3}
\end{eqnarray}
\noindent Proof immediately follows from Lemmas $\ref{lemma2}$, $\ref{lemma3}$. The case $j=0$ is similar.
\end{proof}

\begin{lemma}
\label{lemma4}
Consider again the  ``pseudo-estimators'' $( \ref{eq:e11} )$. Under the assumptions of Lemma $\ref{lemma4}$:
\begin{eqnarray}
\underset{y}{\sup} \left \vert \widetilde{F}_{j,n}(y)-F_j(y) \right \vert
& \overset{a.s.}{\rightarrow} & 0 \;\; {\mathrm{as}} \;  n  \rightarrow \infty , \; j=1, \, 0.
\label{eq:e19}
\end{eqnarray}
\end{lemma}
\begin{proof}[{\bf Proof of Lemma \ref{lemma4}}]
The result can be shown by standard arguments. Consider first the case $j=1$.
From the Strong Law of Large Numbers for {\em i.i.d.} r.v.s, we have:
\begin{eqnarray}
\widetilde{F}_{1,n}(y) \overset{a.s.} {\rightarrow} F_1(y) \;\; {\mathrm{as}} \; n{\rightarrow} \infty , \; \forall \, y \in \ \mathbb{R} .
\label{eq:e12}
\end{eqnarray}
Moreover, on the basis of the properties of $F_1(y)$ (monotone non decreasing, continuous to the left, with total variation equal to 1),
for every positive  $\epsilon$ there exists a partition of $\mathbb{R}$
\begin{eqnarray}
-\infty, < z_0 < z_1 \dots < z_{k-1} < z_k=+\infty \nonumber
\end{eqnarray}
\noindent  such that
\begin{eqnarray}
F_1(z^-_{j+1})-F_1(z_{j})<\frac{\epsilon}{2} \;\; \forall \, j=0, \, 1, \, \dots , \, k-1
\nonumber
\end{eqnarray}
For each $z_{j} < y <z_{j+1}$ it is then:
\begin{eqnarray}
 \widetilde{F}_{1,n}(z_{j} ) \leq \widetilde{F}_{1,n}(y) \leq \widetilde{F}_{1,n}(z^-_{j+1} )  , \;\;\;
 F_1(z_j) \leq F_1(y) \leq F_1(z^-_{j+1})  \nonumber
\end{eqnarray}
\noindent for all $j=0, \, 1, \, \dots, \, k-1$,
and this implies that
\begin{eqnarray}
\widetilde{F}_{1,n}(z_{j} )-F_1(z^-_{j+1}) \leq \widetilde{F}_{1,n}(y)-F_1(y) \leq \widetilde{F}_{1,n}(z^-_{j+1}) -F(z_j)
\;\; \forall \, z_j \le y < z_{j+1} .
\label{eq:e14}
\end{eqnarray}
Moreover, for every $z_j \le y < z_{j+1}$ it is seen that
\begin{eqnarray}
\widetilde{F}_{1,n}(y)-F_1(y) \geq \widetilde{F}_{1,n}(z_j)-F_1(z_j)+F_1(z_j)-F_1(z^-_{j+1})  \geq
\widetilde{F}_{1,n}(z_j)-F_1(z_j)-\frac{\epsilon}{2} ,
\label{eq:e15}
\end{eqnarray}
\noindent and similarly:
\begin{eqnarray}
\widetilde{F}_{1,n}(y)-F_1(y) \leq \widetilde{F}_{1,n}(z^-_{j+1})-F_1(z^-_{j+1}) +F_1(z^-_{j+1}) -F_1(z_j)  \leq
\widetilde{F}_{1,n}(z^-_{j+1})-F_1(z^-_{j+1}) +\frac{\epsilon}{2} .
\label{eq:e16}
\end{eqnarray}
From inequalities $( \ref{eq:e15} )$, $( \ref{eq:e16} )$ it follows that
\begin{eqnarray}
\underset{y}{\sup} |\widetilde{F}_{1,n}(y)-F_1(y)|  \leq \underset{0 \leq j \leq  k-1}{\max} \left \{ \left \vert
\widetilde{F}_{1,n}(z_{j})-F_1(z_j) \right \vert + \left \vert F_1(z^-_{j+1})-F_1(z^-_{j+1}) \right \vert
\right \} + \epsilon
\label{eq:e17}
\end{eqnarray}
As $n \rightarrow \infty$, the Strong Law of Large Numbers implies that
\begin{eqnarray}
 \underset{j}{\max} \left \{ \left \vert \widetilde{F}_{1,n}(z_{j})-F_1(z_j) \right \vert
 + \left \vert F_1(z^-_{j+1})-F_1(z^-_{j+1}) \right \vert \right \} \overset{a.s.}{\rightarrow} 0
\nonumber
\end{eqnarray}
\noindent and since $ \epsilon>0$ can be made arbitrarily small, conclusion $( \ref{eq:e19} )$ follows. The case $j=0$ is dealt with similarly.
\end{proof}

\begin{proof}[{\bf Proof of Proposition \ref{gliv-cant}}]
Immediate consequence of Lemmas $\ref{lemma3}$, $\ref{lemma4}$.
\end{proof}

\begin{proof}[{\bf Proof of Proposition \ref{weak-conv}}]
Using $( \ref{eq:emp_proc_interm2} )$ and the uniform boundedness on compact sets of $y$s of the $o_p (1)$ term, it is enough
to prove that the sequence of stochastic processes
\begin{eqnarray}
\widetilde{W}_n (y) = \left[ \begin{array}{cc}
\widetilde{W}_{1n} (y) \\
\widetilde{W}_{0n} (y)
\end{array} \right]
=
\left[ \begin{array}{cc}
\frac{1}{\sqrt{n}} \sum_{i=1}^{n} Z_{1,i} \\
\frac{1}{\sqrt{n}} \sum_{i=1}^{n} Z_{0,i}
\end{array} \right]
, \;\; y \in \mathbb{R} .
\label{eq:emp_proc_interm3}
\end{eqnarray}
\noindent converges weakly to the Gaussian process $W( \cdot )$. Observing that $E[ \widetilde{W}_{1n} (y) ] =
E[ \widetilde{W}_{0n} (y) ] =0$, and
using Theorem 2.11.1 in \cite{vaartwellner96} (p. 206), we have to prove
point-wise convergence of covariance functions and asymptotic equicontinuity.

\noindent {\bf 1.} {\em Convergence of covariance}. Consider first the term $C( \widetilde{W}_{1n} (y) , \, \widetilde{W}_{1n} (t) )$.
Since $Z_{1,i}$s are {\em i.i.d.} r.v.s, and taking into account that $p_1(x) = p(x)$, we may write
\begin{eqnarray}
C( \widetilde{W}_{1n} (y) , \, \widetilde{W}_{1n} (t) ) & = &
E[W_{1,n}(y ) W_{1,n}(t)] \nonumber \\
\, & = &
E \left [ \left \{ \frac{I_{(T_=1)}}{p_1(x)}(I_{(Y \leq y)}-F_1(y))-\frac{F_1(y \vert x)-F_1 (y )}{p_{1}(x)}(I_{(T_i=1)}-p(x)) \right \} \right .
\nonumber \\
& \, & \left . \left  \{ \frac{I_{(T_i=1)}}{p_1(x)}(I_{(Y \leq t)}-F_1(t))-\frac{F_1 (t \vert x)- F_1(t)}{p_1(x)}(I_{(T_i=1)}-p(x)) \right \} \right ]
\nonumber \\
& = & E \left [ \frac{I_{(T=1)}}{p(x)^2}(I_{(Y_{(1)} \leq y \land t)}-F_1(y) I_{(Y_{(1)} \leq t)}-F_1(t) I_{(Y_{(1)} \leq y)}+F_1(y) F_1(t))
\right ] \nonumber \\
& - & E \left [ \frac{F_1(y \vert x)-F_1(y)}{p(x)^2}I_{(T=1)}(I_{(T=1)}-p(x))(I_{(Y_{(1)} \leq t)}-F_1(t)) \right ] \nonumber \\
&- & E \left [\frac{F_1(t \vert x) - F_1(t)}{p(x)^2}I_{(T=1)}(I_{(T=1)}-p(x))(I_{(Y_{(1)} \le y)}-F_1(y)) \right ] \nonumber \\
& + & E \left [ \frac{F_1(y \vert x)-F_1(y)}{p(x)}\frac{F_1(t \vert x)-F_1(t)}{p(x)}(I_{(T_i=1)}-p(x))^2 \right ] \nonumber \\
& = & E \left [ \frac{1}{p(x)}(F_1{(y \land t \vert x)}-F_1(y \vert x) F_1(t)-F_1(y)F_1{(t \vert x)}+ F_1(y)F_1(t)) \right ] \nonumber \\
& - & E \left [ \left ( \frac{1}{p(x)}-1 \right ) (F_1(y \vert x)-F_1(y))(F_1(t \vert x)-F_1(t))   \right  ] \nonumber \\
& - & E \left [ \left ( \frac{1}{p(x)}-1 \right ) (F_1(y \vert x)-F_1(y))(F_1(t \vert x)-F_1(t))  \right  ] \nonumber \\
& + & E \left [ \left ( \frac{1}{p(x)}-1 \right ) (F_1(y \vert x)-F_1(y))(F_1(t \vert x)-F_1(y))   \right ] \nonumber \\
& = & E \left [ \frac{1}{p(x)} \left ( F_1{(y \land t \vert x)}-F_1(y \vert x)F_1(t \vert x)+(F_1(y \vert x)-F_1(y))
(F_1(t \vert x)-F_1(t)) \right ) \right  ] \nonumber \\
& - & E \left [ \left (\frac{1}{p(x)}-1 \right ) (F_1(y \vert x)- F_1(y))(F_1(t \vert x)-F_1(t)) \right ] \nonumber \\
& = & E_x \left [ \frac{1}{p(x)}(F_1{(y \land t \vert x)}-F_1(y \vert x)F_1(t \vert x)) \right  ] \nonumber \\
& + & E_x \left [(F_1(y \vert x)-F_1(y))(F_1(t \vert x)F_1(t)) \right ] \nonumber \\
& = & C_{11} (y, \, t) , \nonumber
\label{eq:e33}
\end{eqnarray}
\noindent and similarly
\begin{eqnarray}
C( \widetilde{W}_{0n} (y) , \, \widetilde{W}_{0n} (t) ) = C_{00} (y, \, t) . \nonumber
\end{eqnarray}
As far as the cross-covariance terms are concerned, it is immediate to see that $C_{01} (y, \, t) = C_{10} (t, \, y)$. Furthermore:
\begin{eqnarray}
C( \widetilde{W}_{1n} (y) , \, \widetilde{W}_{0n} (t) ) & = & E [ \widetilde{W}_{1n} (y)  \widetilde{W}_{0n} (t) ] \nonumber \\
& = &
E \left [ \left (\frac{I_{(T=1)}}{p(x)}(I_{(Y \leq y)}-F_1(y))-\frac{F_1( y \vert x)-F_1(y)}{p(x)}(I_{(T=1)}-p(x))\right ) \right .
\nonumber \\
& \, & \left . \left  ( \frac{I_{(T=0)}}{1-p(x)}(I_{(Y \leq t)}-F_0(t))-\frac{F_0(t \vert x)-F_0(t)}{1-p(x)}(I_{(T=1)}-p(x))\right )
\right  ] \nonumber \\
& = & E \left [ \left ( \frac{I_{(T=1)}}{p(x)}(I_{(Y_{(1)} \leq y)}-F_1(y \vert x))+(F_1(y \vert x)-F_1(y)) \right ) \right .
\nonumber \\
& \, & \left . \left ( \frac{I_{(T=0)}}{1-p(x)}(I_{(Y_{(0)} \leq t)}-F_0(t \vert x))+(F_0(t \vert x)-F_0(t))
\right ) \right  ] \nonumber \\
& = & E \left [ \frac{I_{(T=1)}}{p(x)}(I_{(Y_{(1)} \leq y)}-F_1(y \vert x))\frac{I_{(T=0)}}{1-p(x)}(I_{(Y_{(0)} \leq t)}
-F_0 (t \vert x)) \right ] \nonumber \\
& + & E \left [ \frac{I_{(T=1)}}{p(x)}(I_{(Y_{(1)} \leq y)}-F_1(y \vert x))(F_0(t \vert x)-F_0(t))\right ]
\nonumber \\
& + & E \left [ \frac{I_{(T=0)}}{1-p(x)}(I_{(Y_{(0)} \leq t)}-F_0 (t \vert x)) (F_1(y \vert x)-F_1(y))
\right ] \nonumber \\
& + & E \left [ (F_1(y \vert x)-F_1(y))(F_0(t \vert x)-F_0(t)) \right ] \nonumber \\
& = & E \left [ (F_1(y \vert x)-F_1(y))(F_0(t \vert x)-F_0(t)) \right ] \nonumber \\
& = & E \left [ F_1(y \vert x) F_0(t \vert x) \right ] - F_1(y)F_0(t) , \nonumber
\end{eqnarray}
\noindent and this ends the ``covariance part'' of the proof.

\noindent {\bf 2.} {\em Asymptotic equicontinuity}. Consider the {\em i.i.d.} r.v.s $( \ref{eq:def-z} )$, and suppose $y <t$. Then:
\begin{eqnarray}
\, & \, & E \left [ \left ( Z_{1,i} (t)- Z_{1,i}(y) \right )^2 \right ] =
\frac{1}{n} \left \{ E[Z_{1,i}(t)^2]+E[Z_{1,i}(y)^2] -2E[Z_{1,i}(t) Z_{1,i} (y) ] \right \} \nonumber \\
& = & \frac{1}{n} \left \{ C_{11}(t,t)+C_{11}(y,y)-2C_{11}(y,t) \right \} \nonumber \\
& = &  \frac{1}{n} \left \{ E \left [ \frac{1}{p(x)} \left ( F_1(t \vert x)
(1-F_1(t \vert x))+F_1(y \vert x)) (1-F_1(y \vert x) \right )
- 2 \left ( F_1(y \vert x)-F_1(t \vert x)F_1(y \vert x) \right ) \right ] \right . \nonumber \\
& + & \left . E \left [ \left ( F_1(t \vert x)-F_1(t) \right )^2 + \left ( F_1(y \vert x)-F_1(y) \right )^2 -
2 \left ( F_1(t \vert x)-F_1(t) \right ) \left ( F_1(y \vert x)-F_1(y) \right ) \right ] \right \} \nonumber  \\
& = & \frac{1}{n} \left \{ E \left [ \frac{1}{p(x)} \left ( F_1 (t \vert x ) - F_1(t \vert x)^2 +
F_1( y \vert x)- F_1 (y \vert x )^2 - 2F_1( y \vert x ) +2 F_1 (y \vert x) F_1(t \vert x) \right ) \right ] \right . \nonumber \\
& + & \left . E \left [ \left ( \left (F_1(t \vert x)-F_1(t) \right ) -
\left ( F_1(y \vert x)-F_1(y) \right ) \right )^2 \right ] \right \} \nonumber \\
& = & \frac{1}{n} \left \{ E \left [ \frac{1}{p(x)}\left ( F_1(t \vert x)-F_1(y \vert x) \right )- \left (
F_1(t \vert x)-F_1(y \vert x) \right )^2 \right ] \right . \nonumber \\
& + & E \left . \left [ \left ( \left ( F_1(t \vert x)-F_1(y \vert x) \right )- \left ( F_1(t)-F_1(y) \right )
\right )^2 \right ] \right \} \nonumber \\
& = & \frac{1}{n} \left \{ E \left [ \frac{1}{p(x)} \left (F_1(t \vert x)-F_1(y \vert x) \right )
\left (1- \left (F_1(t \vert x)-F_1(y \vert x) \right ) \right ) \right ] \right . \nonumber  \\
& + & E \left . \left [ \left ( F_1(t \vert x)-F_1(y \vert x) \right )^2
+ \left ( F_1(t)-F_1(y) \right )^2 -2 \left ( F_1(t \vert x)-F_1(y \vert x) \right ) \left (F_1(t)-F_1(y)
\right ) \right ] \right \} \nonumber  \\
& = & \frac{1}{n} \left \{ E \left [ \frac{1}{p(x)} \left (F_1(t \vert x)-F_1(y \vert x) \right )
\left ( 1-\left (F_1(t \vert x)- F_1(y \vert x) \right ) \right ) \right ] \right . \nonumber  \\
& + & E \left . \left [ \left ( F_1(t \vert x)-F_1(y \vert x) \right )^2 \right ]
- \left (F_1(t)-F_1(y) \right )^2 \right \} \nonumber  \\
& \leq & \frac{1}{n} \left \{ E \left [ \frac{1}{p(x)} \left (F_1(t \vert x)-F_1(y \vert x ) \right ) \right ]
+ E \left [ \left ( F_1(t \vert x)-F_1(y \vert x) \right )^2 \right ] \right \} \nonumber  \\
& \leq & \frac{1}{n} \left \{ \frac{1}{\delta} E \left [F_1(t \vert x)-F_1(y \vert x) \right ] + E \left [ \left (
F_1 (t \vert x)-F_1(y \vert x ) \right ) \right ] \right \} \nonumber \\
& = & \frac{1}{n} \left ( 1+\frac{1}{\delta} \right ) \left (F_1(t)-F_1(y) \right ) .
\label{eq:e35c}
\end{eqnarray}
A similar result is obtained as $t<y$, as well as when $j=0$, so that inequalities:
\begin{eqnarray}
E \left [ \left ( Z_{j,i} (y) - Z_{j,i} (t) \right )^2 \right ] & \le &
\frac{1}{n} \left ( 1+\frac{1}{\delta} \right ) \left \vert F_j (t)- F_j (y) \right \vert , \;\; j=1, \, 0  \label{eq:35d}
\end{eqnarray}
\noindent hold true.

Since $F_j(y)$ is continuous (uniformly, being monotonic and bounded), from
\begin{eqnarray}
\left \vert Z_{1,i} (y) \right \vert & = &
\frac{1}{\sqrt{n}} \left \vert \frac{I_{(T_{i}=j)}}{p_{j} (x_i)} (I_{(Y_{i} \leq y)}-F_j(y)) -
\frac{F_j (y \vert x_i) - F_j (y)}{p(x_i)}(I_{(T_i=j)}-p_{j} (x_i)) \right \vert \nonumber \\
& \leq & \frac{1}{\sqrt{n}} \left \{ \left \vert \frac{1}{p_{j} (x_i)} \right \vert +
\left \vert \frac{1}{p_{j}(x_i)} \right \vert \right \} \nonumber \\
& \leq & \frac{2}{\delta}n^{-\frac{1}{2}} , \;\; j=1, \, 0
\label{eq:e35e}
\end{eqnarray}
it follows that, for every positive $\eta$:
\begin{eqnarray}
\sum_{i=1}^{n} E \left [ \left \| Z_{j,i} (y) \right \|^2
I_{\left (  \| Z_{j,i} (y)  \| > \eta \right )} \right ]
 \leq \frac{4}{\delta^2} I_{\left ( \frac{2}{\delta} > \eta \sqrt{n} \right )}
\rightarrow  0 \;\; {\mathrm{as}} \;  n \rightarrow \infty . \nonumber
\end{eqnarray}

Next, define the (random) pseudometric:
\begin{eqnarray}
d_n(t, \, y) = \sum_{i=1}^{n} \left \{ \left ( Z_{1,i} (t) - Z_{1,i} (y) \right )^2 +
\left ( Z_{0,i} (t) - Z_{0,i} (y) \right )^2 \right \} . \nonumber
\end{eqnarray}
From the Strong Law of Large Numbers it is seen that
\begin{eqnarray}
d_n(t, \, y) & \overset{a.s.}{\rightarrow} & E \left [ \left ( Z_{1,i} (t) - Z_{1,i} (y) \right )^2 +
\left ( Z_{0,i} (t) - Z_{0,i} (y) \right )^2 \right ] \nonumber \\
\, & \leq & c \left ( \left \vert F_1(t) - F_1(y) \right \vert  +
\left \vert F_0(t) - F_0(y) \right \vert \right )
\label{eq:e35f}
\end{eqnarray}
\noindent with $ c =  ( 1+ \delta^{-1} )$.

Denote now by $N(\epsilon,\mathbb{R},d_n)$ the smallest number of intervals of $[y, \, t]$ that cover the real line,
and  such that $d_n(t,y)<\epsilon$ .
By (\ref{eq:e35f}) it follows that, with probability 1, for $n$ large  enough,
\begin{eqnarray}
d_n(t,y) \leq  c \left ( \vert F_1 (t)-F_1 (y) \vert + \vert F_0 (t)-F_0 (y) \vert \right ) . \nonumber
\end{eqnarray}
\noindent Hence, with probability 1, the number $N(\epsilon,\mathbb{R},d_n)$ is bounded by $\frac{K}{\epsilon}$,
$K$ being an appropriate constant. As a consequence, with probability 1, for $n$ large enough, we have:
\begin{eqnarray}
\int_{0}^{h}{\sqrt{\log N(\epsilon,\mathbb{R},d_n)} \, d \epsilon} & \leq & \int_{0}^{h}
{\sqrt{\log\frac{K}{\epsilon}} \, d\epsilon} \nonumber \\
& \leq & K h - \int_{0}^{h} {\sqrt{\log \epsilon} \, d\epsilon}\rightarrow 0 \;\; {\mathrm{as}} \; h \downarrow 0 .
\nonumber
\end{eqnarray}
In view of Theorem 2.11.1 in \cite{vaartwellner96} (p. 206), this completes the proof.
\end{proof}

\begin{proof}[{\bf Proof of Proposition \ref{continuous-trajectories}}]
Let $Q_j(u)=F_j^{-1}(u)  = \inf \{y : \; F_j(y) \geq u \}$, $j=1, \, 0$. Then, $W_j ( \cdot )$ possesses continuous trajectories almost surely if $B_j
(u)=W_1(Q(u))$ possesses continuous trajectories almost surely. From the inequality (consequence of of proof of Proposition $\ref{weak-conv})$:
\begin{eqnarray}
E \left [ (W_1(t)-W_1(y))^2 \right ] \leq c \vert F_j(t) - F_j(y) \vert , \nonumber
\end{eqnarray}
\noindent $c$ being an appropriate constant, it follows that
\begin{eqnarray}
E \left [ (B_j(u)-B_j(v))^2 \right ] \leq c \vert u-v \vert  \;\; \forall \, u, \,v \in (0,1) \label{eq:ineq-lead}
\end{eqnarray}
\noindent
The continuity of the trajectories of $B_j ( \cdot )$ follows from $( \ref{eq:ineq-lead} )$ and formula $(6)$ in
 \cite{leadweis69}.
\end{proof}

\begin{proof}[{\bf Proof of Proposition \ref{asymptotics_wilcoxon}}]
First of all, using an integration by parts we have
\begin{eqnarray}
\widehat{\theta}_{01} - \theta_{01} & = & \int_{\mathbb{R}} \widehat{F}_{0,n}(y) \, d \widehat{F}_{1,n}(y)
- \int_{\mathbb{R}} F_{0}(y) \, dF_{1}(y) \nonumber \\
& = & \int_{\mathbb{R}} \widehat{F}_{0,n}(y) \, d [ \widehat{F}_{1,n}(y) - F_1(y) ]
+ \int_{\mathbb{R}} \left ( \widehat{F}_{0,n}(y) - F_0(y) \right ) \, d F_{1}(y) \nonumber \\
& = & \int_{\mathbb{R}} \left ( \widehat{F}_{0,n}(y) - F_0(y) \right ) \, d [ \widehat{F}_{1,n}(y) - F_1(y) ]
+ \left [
F_0(y) ( \widehat{F}_{1,n}(y) - F_1(y)) \right ]_{-\infty}^{+\infty} \nonumber \\
& - & \int_{\mathbb{R}} \left ( \widehat{F}_{1,n}(y) - F_1(y) \right ) \,  dF_0(y)
+ \int_{\mathbb{R}} \left ( \widehat{F}_{0,n}(y) - F_0(y) \right ) \, dF_1(y) \nonumber
\end{eqnarray}
\noindent and hence
\begin{eqnarray}
\sqrt{n} ( \widehat{\theta}_{01} - \theta_{01} ) = \int_{\mathbb{R}} W_{0,n}(y) \,
d \left [ n^{-1/2} W_{1,n}(y) \right ]
+ \int_{\mathbb{R}} W_{0,n}(y) \, dF_{1}(y)
- \int_{\mathbb{R}} W_{1,n}(y) \, dF_{0}(y)
\label{eq:e44}
\end{eqnarray}
\noindent where $W_{j,n}(y)=\sqrt{n} ( \widehat{F}_{j,n}(y)- F_j (y))$, $j=1$, $0$.

Now, if $F_0(y),F_1(y)$ are continuous, the limiting process $ W=[ W_1, \, W_0]^{\prime}$ possesses trajectories that are continuous (and bounded) with probability 1,
so that it is concentrated on $C( \overline{\mathbb{R}})^2$, that is separable and complete if equipped with the $sup$-norm. Using then  the Skorokhod Representation Theorem (cfr. \cite{billingsley99}, p. 70), there exist processes $ \widetilde{W}_n = [ \widetilde{W}_{1,n} , \, \widetilde{W}_{0,n} ]^{\prime}$, $n \geq 1$, and
$ \widetilde{W} = [ \widetilde{W}_{1} , \, \widetilde{W}_{0} ]^{\prime}$, defined on a probability space
$(\widetilde{\Omega}, \, \widetilde{\mathcal{F}} , \, \widetilde{P})$ such that
\begin{eqnarray}
\widetilde{W}_n \overset{d}{=} W_n \; \forall \, n \geq 1 , \;\; \widetilde{W} \overset{d}{=} W
\label{eq:equal_distr_skor}
\end{eqnarray}
\noindent and
 \begin{eqnarray}
 \sup_y \left \vert \widetilde{W}_{j,n} (y) - \widetilde{W}_{j} (y) \right \vert \rightarrow 0 \;\; {\mathrm{as}} \; n \rightarrow \infty , \;
 j=1, \, 0, \; a.s. - \widetilde{P}
 \label{eq:as_skor}
 \end{eqnarray}
\noindent where the symbol $\overset{d}{=}$ denotes equality in distribution.

From $( \ref{eq:equal_distr_skor}  )$ and $( \ref{eq:e44} )$, the relationship
\begin{eqnarray}
\sqrt{n} ( \widehat{\theta}_{01} - \theta_{01} ) \stackrel{d}{=} \int_{\mathbb{R}} \widetilde{W}_{0,n}(y) \,
d \left [ n^{-1/2} \widetilde{W}_{1,n}(y) \right ]
+ \int_{\mathbb{R}} \widetilde{W}_{0,n}(y) \, dF_{1}(y)
- \int_{\mathbb{R}} \widetilde{W}_{1,n}(y) \, dF_{0}(y)
\label{eq:same_distrib}
\end{eqnarray}
\noindent follows.

The terms appearing in the r.h.s. of $( \ref{eq:same_distrib} )$ can be handled separately. First of all, we have
\begin{eqnarray}
 \int_{\mathbb{R}} \widetilde{W}_{0,n}(y) \, dF_{1}(y) =
 \int_{\mathbb{R}} \left ( \widetilde{W}_{n} (y) - \widetilde{W}_{0}(y) \right ) \, d F_1(y)
 + \int_{\mathbb{R}}  \widetilde{W}_{0}(y) \, d F_1(y) ,
 \nonumber
 \end{eqnarray}
\noindent and since
 \begin{eqnarray}
 \left \vert \int_{\mathbb{R}} \left ( \widetilde{W}_{0,n}(y) - \widetilde{W}_{0}(y) \right ) \, dF_{1}(y) \right \vert
 \leq \sup_y \left \vert \widetilde{W}_{0,n}(y) - \widetilde{W}_{0}(y) \right \vert \rightarrow 0
 \;\; {\mathrm{as}} \; n \rightarrow \infty \;\; {\mathrm{a.s.}} -  \widetilde{P} ,
 \nonumber
 \end{eqnarray}
\noindent we easily obtain
 \begin{eqnarray}
 \int_{\mathbb{R}} \widetilde{W}_{0,n}(y) \, dF_{1}(y) \rightarrow \int_{\mathbb{R}} \widetilde{W}_{0}(y) \, dF_{1}(y)
 \;\; {\mathrm{as}} \; n \rightarrow \infty \;\; {\mathrm{a.s.}} -  \widetilde{P}
 \label{eq:conv1_wilcox}
 \end{eqnarray}
\noindent and similarly
 \begin{eqnarray}
 \int_{\mathbb{R}} \widetilde{W}_{1,n}(y) \, dF_{0}(y) \rightarrow \int_{\mathbb{R}} \widetilde{W}_{1}(y) \, dF_{0}(y)
 \;\; {\mathrm{as}} \; n \rightarrow \infty \;\; {\mathrm{a.s.}} -  \widetilde{P} .
 \label{eq:conv2_wilcox}
 \end{eqnarray}

Finally, for every integer $n$,  $n^{-1/2} \tilde{W}_{1,n}(y)$ is a bounded variation function, with total variation $\leq 2$, a.s.-$\widetilde{P}$, and since
the trajectories of the process $\widetilde{W}_{1}$ are continuous and bounded we may write
\begin{eqnarray}
 n^{-1/2} \widetilde{W}_{1,n}(y) \rightarrow 0 \;\; {\mathrm{as}} \; n \rightarrow \infty , \; a.s. - \widetilde{P} .
 \label{eq:convmis0}
\end{eqnarray}

Relationship $( \ref{eq:convmis0} )$ the signed measure induced by $n^{-1/2} \widetilde{W}_{1,n} $ converges weakly to a measure identically equal to zero.
Hence:
\begin{eqnarray}
 \left \vert
 \int_{\mathbb{R}} \widetilde{W}_{0,n}(y) \, d \left ( n^{-1/2} \widetilde{W}_{1,n}(y) \right ) \right \vert
 & \leq & \left \vert \int_{\mathbb{R}} \widetilde{W}_{0}(y) \, d \left ( n^{-1/2} \widetilde{W}_{1,n}(y) \right ) \right \vert \nonumber \\
 & + & \left \vert \int_{\mathbb{R}} \left ( \widetilde{W}_{0,n}(y) - \widetilde{W}_{0}(y) \right ) \,
 d \left ( n^{-1/2} \widetilde{W}_{1,n}(y) \right ) \right \vert \nonumber \\
 & \leq &  \underset{(a)}{\left \vert \int_{\mathbb{R}} \widetilde{W}_{0}(y) \, d \left ( n^{-1/2} \widetilde{W}_{1,n}(y) \right )
 \right \vert } \nonumber \\
 & + & 2 \underset{(b)}{\sup_{y} \left \vert \widetilde{W}_{0,n}(y) - \widetilde{W}_{0}(y) \right \vert } \nonumber \\
 & \rightarrow & 0  \;\;  {\mathrm{as}} \; n \rightarrow \infty , \; a.s. - \widetilde{P}
 \label{eq:e49}
\end{eqnarray}
\noindent where the term $(a)$ goes to zero according to the Helly-Bray theorem ($\widetilde{W}_0 $ is continuous and bounded a.s. $-\widetilde{P}$),
and the term $(b)$ goes to zero according to the Skorokhod Representation Theorem.

From $( \ref{eq:conv1_wilcox} )$, $( \ref{eq:conv2_wilcox} )$, and $( \ref{eq:e49} )$ it follows that:
\begin{eqnarray}
 & \, & \int_{\mathbb{R}} \widetilde{W}_{0,n}(y) \, d \left( n^{-1/2} \widetilde{W}_{1,n}(y) \right )
 + \int_{\mathbb{R}} \widetilde{W}_{0,n}(y) \, dF_1(y) - \int_{\mathbb{R}} \widetilde{W}_{1,n}(y) \, dF_0(y) \nonumber \\
 & \, & \rightarrow \int_{\mathbb{R}} \widetilde{W}_{0}(y) \, dF_1(y)
 - \int_{\mathbb{R}} \widetilde{W}_{1}(y) \, dF_0(y)  \;\; {\mathrm{as}} \; n \rightarrow \infty , \; a.s. - \widetilde{P}
\end{eqnarray}
\noindent which is equivalent to:
\begin{eqnarray}
\sqrt{n} ( \widehat{\theta}_{01}-\theta_{01} ) \stackrel{d}{\rightarrow}
\int_{\mathbb{R}} W_{0}(y) \, dF_1(y)  - \int_{\mathbb{R}} W_{1}(y) \, dF_0(y)  \;\; {\mathrm{as}} \; n \rightarrow \infty  .
\label{eq:e51}
\end{eqnarray}

The r.h.s. of $( \ref{eq:e51} )$ is a linear functional of a Gaussian process with continuous and bounded trajectories,
so that it possesses Gaussian distribution with zero expectation and variance
\begin{eqnarray}
V = V_1+V_2-2V_3
\label{eq:e52}
\end{eqnarray}
\noindent where
\begin{eqnarray}
V_1 & = &
\int_{\mathbb{R}^2} E[ W_{0}(y) \, W_{0}(t)] \, dF_1(y) \, dF_1(t) \label{eq:V1} ,  \\
V_2 & = & \int_{\mathbb{R}^2} E[ W_{1}(y) \, W_{1}(t)] \, dF_0(y) \, dF_0(t) ,  \label{eq:V2} \\
V_3 & = & \int_{\mathbb{R}^2} E[ W_{0}(y) \, W_{1}(t)] \, dF_1(y) \, dF_0(t) . \label{eq:V3}
\end{eqnarray}

The terms $V_1$ - $V_3$ in $( \ref{eq:V1} )$ - $( \ref{eq:V3} )$ can be written more compactly. Using the quantities
$\gamma_{10} (x)$, $\gamma_{01} (x)$ defined in $( \ref{eq:def_gamma} )$, it is not difficult to see that
\begin{eqnarray}
V_1 & = & \int_{\mathbb{R}^2} E_x \left [ \frac{1}{p(x)} (F_1 (y \land t \vert x ) - F_1 (y \vert x) \, F_1 (t \vert x) )
\, dF_0(y) \, dF_0(t) \right ] \nonumber  \\
 & + &
\int_{\mathbb{R}^2} E_x \left [ (F_1(y \vert x) - F_1(y)) (F_1( t \vert x ) - F_1(t)) \right ] \, dF_0(y) \, dF_0(t) \nonumber \\
 & = &
 E_x \left [ \frac{1}{p(x)} \left \{ \int_{\mathbb{R}^2} \left ( E \left [ \left . \frac{I_{(T=1)}}{p(x)}I_{(Y \leq y\land t)} 
\right  \vert x \right ]
 \, dF_0(y) \, dF_0(t) \right ) \right . \right . \nonumber \\
 & - & \left .  \left . \left ( \int_{\mathbb{R}} \left ( E \left [ \left . \frac{I_{T=1}}{p(x)}I_{(Y \leq y)} \right \vert x \right ]
 \right ) \, dF_0(y) \right )^2 \right \} \right ]
 +  E_x[(\gamma_{10}(x)-\theta_{10})^2] \nonumber \\
 & = &
E_x \left [ \frac{1}{p(x)} \left \{ E \left [  \frac{I_{(T=1)}}{p(x)}
 \int_{\mathbb{R}^2} I_{( y \land t \ge Y)} \, dF_0(y) \, dF_0(t) \vert x \right ]
 - \left ( E \left [ \frac{I_{(T=1)}}{p(x)} \int_{\mathbb{R}} I_{(y \ge Y)} dF_0(y) \right ] \right )^2
 \right \} \right ] \nonumber \\
 & + & V_x ( \gamma_{10}(x)) \nonumber \\
& = &
 E_x \left [ \frac{1}{p(x)} \left \{ E \left [ \left (1-F_0(Y_1) \right )^2 \vert x \right ]
 - \left ( E \left [1-F_0(Y_1) \vert x \right ] \right )^2 \right \} \right ] + V_x(\gamma_{10}(x)) \nonumber \\
 & = &
E_x \left [ \frac{1}{p(x)} V \left (F_0(Y_1) \vert x \right ) \right ] + V_x(\gamma_{10}(x))
\label{eq:e53}
\end{eqnarray}
In the same way, it is seen that:
\begin{eqnarray}
V_2 = E_x \left [ \frac{1}{1-p(x)} V \left ( F_1(Y_0) \vert x \right ) \right ] + V_x(\gamma_{01}(x))
\label{eq:e53b}
\end{eqnarray}
\noindent and
\begin{eqnarray}
V_3 & = & \int_{\mathbb{R}^2} E \left [ W_0(y) \, W_1(t) \right ] \, dF_1(y) \, dF_0(t)
\nonumber \\
& = & E_x \left [ (\gamma_{10}(x)-\theta_{10}) (\gamma_{01}(x)-\theta_{01}) \right ]
\label{eq:e54}
\end{eqnarray}
From $( \ref{eq:e53} )$ - $( \ref{eq:e54} )$, $( \ref{eq:e55} )$ easily follows.
\end{proof}

\end{document}